\documentclass[11pt]{article}

\usepackage{amsmath,amsfonts,amssymb,amsthm}
\usepackage{amsthm}
\usepackage{graphicx,color}
\usepackage{framed}
\usepackage{enumitem}
\usepackage{thmtools}
\usepackage{thm-restate}
\usepackage{xspace}
\usepackage{bm}
\usepackage{todonotes}
\usepackage{footnote}
\usepackage{mathtools}
\usepackage{mathrsfs}
\usepackage[margin=1in]{geometry}
\usepackage{todonotes}
\usepackage[most]{tcolorbox}
\usepackage{footnote}
\usepackage{rotating}
\usepackage[pdftex, plainpages = false, pdfpagelabels, 
bookmarks=true,
bookmarksopen = true,
bookmarksnumbered = true,
breaklinks = true,
linktocpage,
pagebackref,
colorlinks = true,  
linkcolor = blue,
urlcolor  = blue,
citecolor = red,
anchorcolor = green,
hyperindex = true,
hyperfigures
]{hyperref} 
\usepackage{tabularx}
\usepackage{tikz} 
\usetikzlibrary{calc}
\usepackage{thm-autoref}
\usepackage[nameinlink]{cleveref}

\newtheorem{theorem}{Theorem}
\newtheorem{lemma}{Lemma}
\newtheorem{claim}{Claim}
\newtheorem{corollary}{Corollary}
\newtheorem{definition}{Definition}
\newtheorem{observation}{Observation}
\newtheorem{proposition}{Proposition}

\newtheorem{remark}{Remark}

\theoremstyle{definition}

\theoremstyle{definition}
\newtheorem{rrule}{Reduction Rule}

\newcommand{\cwB}{\widetilde{\mathcal{B}}}

\newcommand{\chS}{\widehat{\mathcal{S}}}
\newcommand{\wU}{\widetilde{U}}

\newcommand{\wB}{\widetilde{B}}

\newcommand{\I}{\cal{I}}
\newcommand{\mat}{$M=(E,{\cal I})$}

\newcommand{\whnd}[1]{\widehat{#1}}

\newcommand{\rep}[2] {$\widehat{{\cal #1}} \subseteq_{rep}^{#2} {\cal #1}$}

\newcommand{\bnoml}[2]{  $\binom{{#1}}{{#2}}$}

\newcommand{\repmat}[1]{$A_{#1}$}

\newcommand{\tgemrand}{$\cO\left({p+\ell \choose p} t p^\omega + t {p+\ell \choose p} ^{\omega-1} + \vert \vert A_M \vert \vert^{\cO(1)} \right)$}
\newcommand{\cO}{{\mathcal{O}}}
\newcommand{\nab}{\mathop{\triangledown}}

\crefname{invar}{invariant}{invariants}
\crefname{ineq}{inequality}{inequalities}
\crefname{constr}{constraint}{constraints}
\crefname{tbl}{table}{tables}
\crefname{lem}{lemma}{lemmata}
\crefname{lemma}{lemma}{lemmata}
\crefname{cond}{condition}{conditions}
\crefname{rrule}{reduction rule}{reduction rules}

\title{Fixed-Parameter Algorithms for Fair Hitting Set Problems}

\author{
	Tanmay Inamdar\thanks{
		Department of Informatics, University of Bergen, Norway.}
	\and
	Lawqueen Kanesh\thanks{Indian Institute of Technology, Jodhpur}
	\and
	Madhumita Kundu\addtocounter{footnote}{-2}\footnotemark{}
	\and
	Nidhi Purohit\addtocounter{footnote}{-1}\footnotemark{} 
	\and
	Saket Saurabh\addtocounter{footnote}{-1}\footnotemark{} \addtocounter{footnote}{1}\thanks{Institute of Mathematical Sciences, Chennai. 
	\\This research is supported by the European Research Council (ERC) under the European Union’s Horizon 2020 research and innovation programme (grant agreement No. 819416) Swarnajayanti Fellowship (No. DST/SJF/MSA01/2017-18).} 
}

\date{}

\newcommand{\lr}[1]{\left(#1\right)}
\newcommand{\LR}[1]{\left\{#1\right\}}

\newcommand{\nat}{\mathbb{N}}

\newcommand{\gray}[1]{\iffalse{\color{lightgray}#1}\fi}

\newcommand{\cU}{\mathcal{U}}
\newcommand{\cF}{\mathcal{F}}
\newcommand{\cB}{\mathcal{B}}
\newcommand{\cM}{\mathcal{M}}
\newcommand{\cI}{\mathcal{I}}
\newcommand{\cS}{\mathcal{S}}
\newcommand{\cR}{\mathcal{R}}
\newcommand{\cH}{\mathcal{H}}
\newcommand{\Oh}{\mathcal{O}}
\newcommand{\cG}{\mathcal{G}}
\newcommand{\cY}{\mathcal{Y}}
\newcommand{\np}{\textsf{NP}}

\newcommand{\tw}{\textup{\texttt{tw}}}

\newcommand{\bbF}{\mathbb{F}}

\newcommand{\cP}{\mathcal{P}}

\newcommand{\FO}{\textup{\textsf{FO}}\xspace}

\newcommand{\chs}{\textup{\textsc{Fair Hitting Set}}\xspace}
\newcommand{\hs}{\textup{\textsc{Hitting Set}}\xspace}
\newcommand{\kmci}{\textup{\textsc{$k$-Multicolored Independent Set}}\xspace}
\newcommand{\shs}{\textup{\textsc{Sparse Hitting Set}}\xspace}

\begin{document}

\maketitle

\begin{abstract}
	Selection of a group of representatives satisfying certain fairness constraints, is a commonly occurring scenario. Motivated by this, we initiate a systematic algorithmic study of a \emph{fair} version of \textsc{Hitting Set}. In the classical \textsc{Hitting Set} problem, the input is a universe $\cU$, a family $\cF$ of subsets of $\cU$, and a non-negative integer $k$. The goal is to determine whether there exists a subset $S \subseteq \cU$ of size $k$ that \emph{hits} (i.e., intersects) every set in $\cF$. Inspired by several recent works, we formulate a fair version of this problem, as follows. The input additionally contains a family $\cB$ of subsets of $\cU$, where each subset in $\cB$ can be thought of as the group of elements of the same \emph{type}. We want to find a set $S \subseteq \cU$ of size $k$ that (i) hits all sets of $\cF$, and (ii) does not contain \emph{too many} elements of each type. We call this problem \chs, and chart out its tractability boundary from both classical as well as multivariate perspective. Our results use a multitude of techniques from parameterized complexity including classical to advanced tools, such as, methods of representative sets for matroids, \FO model checking, and a generalization of best known kernels for \textsc{Hitting Set}. 
\end{abstract}

\section{Introduction} \label{sec:intro} 
Imagine a scenario of selecting a committee of size $k$ from a group of people $\cU$. We need a committee of people with some given attributes. These kinds of ``attribute hitting'' scenarios is modeled by a family $\cF$ over $\cU$, where for each attribute $\mathscr{A}$ , we have a  set $\cF$ containing people in $\cU$ who have the attribute $\mathscr{A}$. As is life, not always every set of people can work collectively. In particular, the committee cannot operate smoothly if we select more than the desired number of people from a set $B \subseteq \cU$. These conflicts are modeled by another family, $\cB$ over $\cU$, and a function $f: \cB \to \nat$, which says that $f(B)$ is the maximum number of people from a set $B \in \cB$ that can serve on the committee. Specifically, we want a committee that is a hitting set for attributes and has a set of people who are ``conflict free''. This paper aims to undertake a systematic study of a generalization of Hitting Set, which models such scenarios, and study this problem in the realm of parameterized complexity.

Indeed, \hs is one of the 21 problems proven to be \np-complete by Karp~\cite{karp1972}.
Recall, in this problem, we are given a set system $(\cU, \cF)$, and an integer $k$. 
Here, $\cU$ is a finite set of elements known as \emph{universe} and $\cF$ is a family of subsets of $\cU$. 
The objective is to determine whether there exists a subset $S \subseteq \cU$ such that $S$ \emph{hits} all sets in $\cF$, i.e., for every $F_i \in \cF$, $S \cap F_i \neq \emptyset$.
\hs is closely related to the \textsc{Set Cover} problem. These two problems, along with a particularly interesting special case thereof, namely that of \textsc{Vertex Cover}, are some of the most extensively studied problems in the field of approximation algorithms and parameterized complexity. \hs problem is of particular interest, because many combinatorial problems can be modeled as instances of \hs.

Motivated from real-life applications, there has been a growing interest on the fairness aspect of various problems and algorithms developed. This has led to the whole new field of algorithmic fairness. Depending on the specific application, there are numerous ways to define the notion of \emph{fairness}. One of the earliest definitions of fairness comes from \cite{LinS89a}, who defined fair versions of edge deletion problems. This was motivated from the following scenario. Suppose the graph models a communication network, with each edge being a link between a pair of nodes. In order to achieve acyclicity in the network, some links need to be disconnected. However, from the perspective of each node, it is desirable that fewest possible links incident to it are disconnected. Thus, we wish to disconnect links in a fair or equitable manner for the nodes.

Subsequently, this notion was extended by \cite{MasarikT20,KnopMT19} to define fair versions of vertex deletion problems. In this model, we want to delete a subset of vertices in order to achieve a certain graph property, such that each vertex has fewest possible neighbors deleted. As a concrete example, in a fair version of \textsc{Vertex Cover} in this model, we want to find a vertex cover $S$, such that each vertex outside $S$ has fewest neighbors in $S$. Recently, \cite{BlumDFGP22} studied a generalization of this, called \textsc{Sparse Hitting Set}. The input to \shs consists of $(\cU, \cF, \cB)$, where $\cU$ is the universe, and $\cF$ and $\cB$ are two families of subsets of $\cU$. The goal is to find a hitting set $S \subseteq \cU$ for $\cF$ such that $k \coloneqq \max_{B_i \in \cB}|B_i \cap S|$ is minimized. Here, $k$ is called the \emph{sparseness} of the solution. Note that \shs generalizes \textsc{Fair Vertex Cover} as defined above. Along a similar line, \cite{JainKM20} considered \emph{conflict-free} versions of various problems, including \hs. In \textsc{Conflict Free $d$-Hitting Set}, we are given an instance $(\cU, \cF, k)$ of \hs, and a conflict graph $H = (\cU, E)$, and the goal is to find a hitting set $S \subseteq \cU$ of size at most $k$, such that $S$ induces an independent set in the conflict graph $H$.

\paragraph*{Our Problem.}
Along the same line of work, we define a fair version of \hs, which captures all of the aforementioned problems, and much more. Formally, the problem is defined as follows.

\begin{tcolorbox}[colback=white!5!white,colframe=gray!75!black]
	\chs
	\\\textbf{Input.} An instance $\cI = (\cU, \cF, \cB, f: \cB \to \nat, k)$, where
	$\cU = \{u_1, u_2, \ldots, u_n\}$ is the universe; $\cB$ and $\cF$ are two families of subsets of $\cU$, where $\cF = \{F_1, F_2, \ldots, F_m\}$, and $\cB = \{B_1, B_2, \ldots, B_\ell\}$, and $k$ is a positive integer.
	\\\textbf{Task.} Determine whether there exists $S \subseteq \cU$, with the following properties.
	\begin{itemize}
		\item $|S| \le k$,
		\item $S$ is a \emph{hitting set} for $\cF$, i.e., for every $F_i \in \cF$, $S \cap F_i \neq \emptyset$, and
		\item For every $B_j \in \cB$, $|S \cap B_j| \le f(B_j)$.
	\end{itemize}
	We refer to a set $S \subseteq \cU$ satisfying the above properties as a \emph{fair hitting set} for  $\cF$, and use $|\cI|$ to denote the size of the instance $\cI$.
\end{tcolorbox}

We note that \chs generalizes \textsc{Sparse Hitting Set}.
Given an instance $(\cU, \cF, \cB)$ of \shs, we iteratively solve instances $\cI_i$ of \chs for $i = 1, 2, \ldots$. Here, an instance $\cI_i$ of \chs is given by $(\cU, \cF, \cB, f_i, |\cU|)$, where $f_i(B_j) = i$ for all $B_j \in \cB$. For the smallest $i$ such that $\cI_i$ is a yes-instance of \chs, we stop and conclude that $i$ is the optimal sparseness of the given instance of \shs. We note that \chs also generalizes the setting considered by \cite{JainKM20}. 

\subsection{Our Results, Techniques, and Relation to Hitting Set}
First, we observe that \chs is a generalization of \hs, by setting $\cB = \emptyset$. Thus, \chs inherits all lower bound results from \hs, namely, in general the problem is \np-hard as well as W$[2]$-hard parameterized by $k$, the solution size~\cite{CyganFKLMPPS15}. 
However, note that in the hard instances of \hs, the sets in $\cF$ can intersect arbitrarily. Indeed, consider an extreme case, when the sets in $\cF$ are pairwise disjoint. In this setting \hs is trivial to solve -- an optimal solution must contain exactly one element from each set of $\cF$. In contrast, we show that \chs remains \np-hard, as well as W$[1]$-hard w.r.t. $k$---and thus unlikely to be FPT---even in this simple setting. In particular, we show the following lower bound results that are proved in \Cref{sec:lb}.
\begin{theorem}
	\label{intro:themHard}
	\chs remains \np-hard when (1) the sets in $\cF$ are pairwise disjoint, (2) each element appears in at most two distinct $B_i$'s in $\cB$, and (3) each $B_i \in \cB$ has size exactly $2$. Furthermore, assuming ETH, it is not possible to solve \chs in time $2^{o(t)}$, where $t = \max\{|\cU|, |\cF|, |\cB|\}$.. 
	
	\chs is W$[1]$-hard when parameterized by $k$, even when the sets in $\cF$ are pairwise disjoint, and each $B_i \in \cB$ has size exactly $2$. 
\end{theorem}
The first result is obtained via a reduction from a problem of finding a ``rainbow matching'' on a path, and for the second result we give a parameter preserving reduction from \kmci. 
Given these lower bound results (Theorem~\ref{intro:themHard}), we study \chs under specific assumptions on the instance $\cI = (\cU, \cF, \cB, f, k)$.  A natural question is: {\em under which assumptions?} To answer this we look at the known fixed-parameter tractability results for \hs. 

\medskip\noindent\textbf{Hitting Set in Parameterized Complexity.}
\hs is known to be W$[2]$-complete parameterized by the solution size in general. In other words, under widely believed complexity theoretic assumptions, it does not admit an FPT algorithm parameterized by the solution size. This motivates the study of \hs in special cases.
One particularly interesting case is \textsc{Vertex Cover} when the size of each set in $\cF$ is exactly two. 
\textsc{Vertex Cover} is the most extensively studied problem in the parameterized complexity with a number of results in the FPT  algorithms and kernelization in general graphs as well as special classes of graphs. 
Many of the techniques and results developed for \textsc{Vertex Cover} also extend $d$-\hs, where each set in $\cF$ has size at most $d$, for some constant $d$. More generally, \hs is known to be FPT and admits a polynomial kernel in the case when the incidence graph $G_{\cU, \cF}$, which is the bipartite graph on the vertex set $\cU\uplus \cF$ with edges denoting the set-containment, is $K_{i,j}$-free. That is, no $i$ sets in $\cF$ contain $j$ elements in common, where $i$ and $j$ are assumed to be constants. This setting generalizes all the above settings as well as when the $G_{\cU, \cF}$ is $d$-degenerate (since such graphs are $K_{d+1,d+1}$-free).

\medskip\noindent\textbf{Our Algorithmic Results.} Notably, we are able to extend almost all of the fixed-parameter tractability results for \hs mentioned in the previous paragraph, under suitable assumptions on the set system $(\cU, \cB)$. We give a summary of our results in \Cref{figure:1}. 

More specifically, we obtain our results in the following steps. Consider a special case \chs, when the sets in $\cF$ are pairwise disjoint, and each element appears in at most $q$ sets in $\cB$. Note that the first part of \Cref{intro:themHard} implies that the problem is \np-hard even when $q = 2$. On the other hand, when $q = 1$, i.e., when both $\cF$ and $\cB$ are families of pairwise disjoint sets, then we observe \chs can be solved in polynomial time. Thus, $q = 1$ to $2$ is a sharp transition between the tractability of the problem. Although the problem is \np-hard even for constant values of $q$, the following results are interesting in this setting. In particular, we show that the problem is FPT, and admits a polynomial kernel parameterized by $k$, if $q$ is a constant.

\begin{restatable}{theorem}{repsetsimple}
	\label{thm:RepSetSimple}
	Let $(\cU, \cF, \cB, f: \cB \to \nat, k)$ be an instance of \chs. Then, 
	\chs can be solved in time $2^{\cO(qk)}n^{\cO(1)}$ time, when every element in $\cU$ appears in at most $q$ sets in $\cB$ and any pair of sets in $\cF$ are pairwise disjoint.  Further, \chs admits a kernel of size 
	$\cO(kq^2 {kq \choose q} \log k)$. 
\end{restatable}

Next we generalize Theorem~\ref{thm:RepSetSimple} to a scenario where  every element in $\cU$ appears in at most $q$ sets in $\cB$ and at most $d$ sets in $\cF$.

\begin{restatable}{theorem}{repsetsimpleplus}
	\label{thm:RepSetSimplePlus}
	Let $(\cU, \cF, \cB, f: \cB \to \nat, k)$ be an instance of \chs. Then, 
	\chs can be solved in time $k^{\cO(dk)}2^{\cO(qk)}n^{\cO(1)}$ time, when every element in $\cU$ appears in at most $q$ sets in $\cB$ and at most $d$ sets in $\cF$. 
\end{restatable}

These results, Theorems~\ref{thm:RepSetSimple} and \ref{thm:RepSetSimplePlus}, are obtained by the key observation that the problem can be modeled as finding a hitting set for $\cF$ that is also an independent set in a suitably defined partition matroid that encodes the constraints imposed by $(\cU, \cB, f)$. This enables us to use the representative sets toolkit developed for matroids. This result is discussed in \Cref{sec:repsets}.

Next we  consider a generalization of the above setting, where (1) each element appears in at most $q$ sets in $\cB$, and (2)  the $G_{\cU, \cF}$ is $K_{d, d}$-free. In this case, we combine the techniques developed in \hs literature in the $K_{d, d}$-free setting, as well as, the representative sets based techniques developed in \Cref{sec:repsets}, to obtain the following result.

\begin{restatable}{theorem}{thmKdd}
	\label{thm:thmKdd}
	Given an instance $\cI = (\cU, \cF, \cB, f, k)$ of \chs, such that $G_{\cU, \cF}$ is $K_{d, d}$-free, and the frequency of each element in $\cB$ is bounded by $q$, one can find an equivalent instance $\cI' = (\cU', \cF', \cB', f', k')$ of \chs in polynomial time, such that $|\cU'| = \Oh(k^{d^2+q}d^{d}q^q)$, $|\cF'| \le dk^d$, and $|\cB'| = \Oh(k^{d^2 + q} \cdot d^{d} q^{q+1})$, where $d$ and $q$ are assumed to be constants.
\end{restatable}

Finally, we reach our most general case, where suppose (1) the $(\cU, \cB)$ incidence graph is ``nowhere dense'' (defined formally in \Cref{subsec:nowhere-prelims}; this class includes planar, excluded minor, bounded degree, and bounded expansion graphs), and (2) the $(\cU, \cF)$ incidence graph is $K_{d, d}$-free. In this case, we obtain an FPT algorithm, parameterized by $k$ and $d$. This result is in two steps. First, we proceed as prior to the case when each element appears in $f(k, d)$ sets of $\cF$ (cf. \Cref{thm:thmKdd}). Next, since $(\cU, \cB)$ incidence graph is nowhere dense, we reduce the problem of finding a \chs to \FO model checking procedure on nowhere dense graphs, which is known to be FPT in the size of the formula. In particular, we show that the problem can be encoded by a variant of {\sc Induced Subgraph Isomorphism} on nowhere dense graphs, where the size of the host graph we are searching for can be bounded by a function of $k$, $d$ and the graph class.

\begin{restatable}{theorem}{nowhere} \label{thm:nowhere-dense}
	Let $\cG$ be a nowhere dense graph class. Let $\cI = (\cU, \cF, \cB, f, k)$ be an instance of \chs such that the incidence graph $G \coloneqq G_{\cU, \cB} \in \cG$, and $G_{\cU, \cF}$ is $K_{d, d}$-free for some $d \ge 1$. Then, one can solve \chs on $\cI$ in time $h(k, d) \cdot |\cI|^{\Oh(1)}$, for some function $h(\cdot, \cdot)$. 
\end{restatable}

\begin{figure}
	\centering
		\begin{tabular} { |p{0.3cm}|p{4cm}|p{4cm}|p{4cm}|  }
			\hline
			No. & $\cG_{\cU, \cB}$ &	$\cG_{\cU, \cF}$ & Results    \\
			\hline
			$1$	& 	$q = 1$  & $d=1$  & Polynomial time \\
			& 	($B_i's$\text{\ are disjoint}) & ($F_i's$ \text{\ are disjoint}) &   \\
			\hline
			$2$	& 	$q =2$    &  $d=1$ & NP-Hard  \\	
			&    &       & No sub-exp algo(ETH)  \\
			
			\hline
			$3$		& 	$q$ & $d=1$ & $2^{\Oh(qk)} \cdot |\cI|^{\Oh(1)}$(\Cref{thm:RepSetSimple})\\
			& 	 &       &   \\
			&  & &	 $\Oh(kq^{2} \binom{kq}{q}\log k)$(kernel) \\
			& 	  &       & (Theorem~\ref{thm:kernel})  \\
			
			\hline
			$4$	& 	$q$  & $d$   &      $k^{\Oh(dk)}2^{qk}|\cI|^{\Oh(1)}$\\
			& 	     &       & (Theorem~\ref{thm:RepSetSimplePlus})  \\ 
			\hline
			$5$	& 	$q$  &  $K_{d,d}$-free   & $k^{\Oh(d^{2}+q)} d^{d}  q^{q+1}$(kernel)  \\
			& 	  &       & (Theorem~\ref{thm:thmKdd}) \\ 
			\hline
			$6$	& 	apex-minor free  &  $K_{d,d}$-free   & FPT$/(k+d)$  (\Cref{thm:tw-apex-minorfree}) \\
			\hline
			$7$	& 	nowhere dense   &  $K_{d, d}$-free &  FPT$/(k+d)$ (Theorem~\ref{thm:nowhere-dense}) \\
			\hline
			$8$		&    $K_{2,2}$-free    & $d=1$& $W[1]$-Hard$/k$ (Theorem~\ref{thm:w1hard})\\
			\hline
		\end{tabular}
	\caption{An overview of different results obtained in this paper. In the second (resp.\ third) column, we state the assumption on the set system $(\cU, \cB)$ (resp.\ $(\cU, \cF)$). In rows 1-5 (resp.\ rows 1-4) $q$ (resp.\ $d$) denotes the maximum frequency of an element in $\cB$ (in $\cF$). In the last column, we mention our results in the respective settings, and give corresponding references. }
	\label{figure:1}
\end{figure}

\begin{figure}
	\centering
		\includegraphics[scale=0.4]{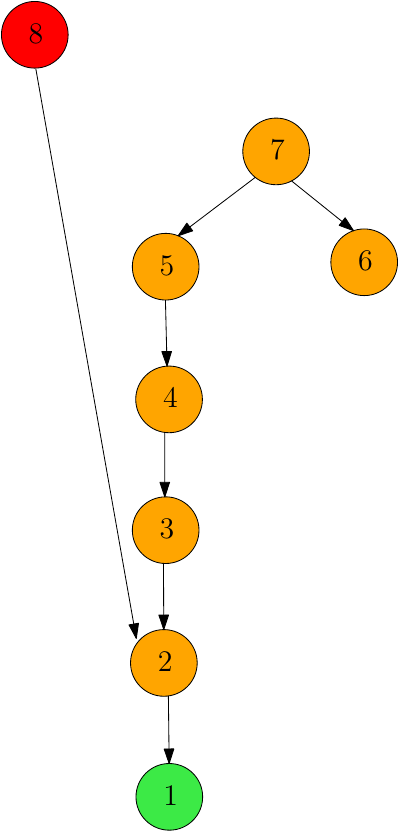}
		\caption{A Hasse diagram of the settings considered in \Cref{figure:1}, where the number in each node corresponds to the row in the table. An arrow from node $i$ to node $j$ indicates that the setting in row $i$ generalizes the setting in row $j$. Nodes colored in green, orange (resp. red) color indicate that the setting is solvable in polynomial, FPT time (resp. is $W[1]$-hard).} \label{figure:2}
\end{figure}
\section{Preliminaries} \label{sec:prelims}
For an integer $\ell \ge 1$, we use the notation $[\ell] \coloneqq \{1, 2, \ldots, \ell\}$. Let $\cR=(\cU, \cS)$ be a \emph{set system}, where $\cU$ is a finite set of elements (also called the \emph{ground set} or the \emph{universe}), and $\cS$ is a family of subsets of $\cU$. For an element $u \in \cU$, and any $\cS' \subseteq \cS$, we use the notation $\cS'(u) \coloneqq \{S \in \cS': u \in S\}$, i.e., $\cS'(u)$ is the sub-family of sets from $\cS'$ that contain $u$. For a subset $R \subseteq \cU$, we denote $\cS - R \coloneqq \{ S \setminus R : S \in \cS \}$. We use $G_{\cU, \cS}$ to denote the incidence graph corresponding to the set system $(\cU, \cS)$, i.e., $G_{\cU, \cS}$ is a bipartite graph with bipartition $\cU \uplus \cS$, such that there is an edge between an element $e \in \cU$ and a set $S \in \cS$ iff $e \in S$. 

In this paper, we work with finite, simple, undirected graphs. We use the standard graph theoretic notation and terminology, as defined in \cite{DiestelGT}. In the following subsections, we give additional preliminaries.

\subsection{Preliminaries on Tree Decompositions} \label{subsec:tw-prelims}
\noindent\textbf{Tree Decompositions and Treewidth}
A \emph{tree decomposition} of a graph $G$ is a pair $\mathcal{T} = (T, \beta)$, where $T$ is a tree, and $\beta: V(T) \to 2^{V(G)}$ satisfying the following properties. 
\begin{itemize}
	\item $\bigcup_{t \in V(T)} \beta(t) = V(G)$.
	\item For every edge $\{u, v\} \in V(G)$, there exists a $t \in V(T)$, such that $\{u, v\} \subseteq \beta(t)$.
	\item For every $u \in V(G)$, the set $T_u = \{t \in V(T): u \in \beta(t) \}$ induces a connected subtree of $T$.
\end{itemize}
We refer to vertices of $T$ as \emph{nodes} in order to distinguish them vertices of $G$, and the set $\beta(t)$ is called the \emph{bag} of $t$. The \emph{width} of a tree decomposition $(T, \beta)$ is defined to be $\max_{t \in V(T)} |\beta(t)| - 1$. The \emph{treewidth} of a graph $G$ is the minimum possible width of a tree decomposition of $G$, which we denote by $\tw(G)$. 

\noindent\textbf{Nice Tree Decompositions.} It is often convenient to deal with a more structured tree decomposition, called \emph{nice tree decomposition}. We think of tree decomposition $(T, \beta)$ as a rooted tree, with a distinguished node $r \in V(T)$ called the \emph{root}. 
Furthermore, the tree decomposition satisfies the following properties.
\begin{itemize}
	\item $\beta(r) = \emptyset$.
	\item Every non-leaf node of $T$ is of one of following three types:
	\begin{itemize}
		\item \textbf{Introduce node.} A node $t$ with exactly one child $t'$ such that $\beta(t) = \beta(t') \cup \{v\}$ for some $v \not\in \beta(t')$. We say that $v$ is \emph{introduced at $t$}.
		\item \textbf{Forget node.} A node $t$ with exactly one child $t'$ such that $\beta(t) = \beta(t') \setminus \{w\}$ for some $w \in \beta(t')$. We say that $w$ \emph{is forgotten} at $t$.
		\item \textbf{Join node.} A node $t$ with exactly two children $t_1, t_2$ such that $\beta(t) = \beta(t_1) = \beta(t_2)$. 
	\end{itemize} 
\end{itemize}
It is known that any tree decomposition of $G$ can be converted into a nice tree decomposition of $G$ of the same width in polynomial time.

For a node $t \in V(T)$, all nodes $t'$ in the subtree rooted at $t$ (including $t$ itself) are called \emph{descendants} of $t$. We note that in a nice tree decomposition $(T, \beta)$, if $v \in V(G)$ is forgotten at $t \in V(T)$, then it must be introduced at a descendant $t'$ of $t$, such that $t' \neq t$.

\subsection{Preliminaries for Nowhere Dense Graphs} \label{subsec:nowhere-prelims}
The following definitions pertaining to the class of nowhere dense graphs are taken verbatim from \cite{LokshtanovPSSZ20}. 

\begin{definition}[Shallow minor]
	A graph~$M$ is an \emph{$r$-shallow minor} of~$G$, where~$r$ is an
	integer, if there exists a set of disjoint subsets $V_1, \ldots,
	V_{|M|}$ of~$V(G)$ such that
	\begin{enumerate}
		\item Each graph $G[V_i]$ is connected and has radius at most~$r$,
		and
		\item There is a bijection $\psi \colon V(M) \to \{V_1, \ldots,
		V_{|M|}\}$ such that for every edge $uv \in E(M)$ there is an edge
		in~$G$ with one endpoint in $\psi(u)$ and second in $\psi(v)$.
	\end{enumerate}
	The set of all $r$-shallow minors of a graph~$G$ is denoted by~$G
	\nab r$.  Similarly, the set of all $r$-shallow minors of all the
	members of a graph class $\cal G$ is denoted by 
	${\cal G} \nab r =\bigcup_{G \in {\cal G}} (G \nab r)$.
\end{definition}
We first introduce the definition of a nowhere dense graph class; let $\omega(G)$ denote the size of the largest clique in $G$ and $\omega({\cal G})=\sup_{G\in \cal G} \omega(G)$.

\begin{definition}[Nowhere dense]\label{def:nd}
	A graph class $\cal G$ is \emph{nowhere dense} if there exists a
	function $f_\omega \colon \mathbb{N} \to \mathbb{N}$ such that for
	all~$r$ we have that $\omega({\cal G} \nab r) \leq f_\omega(r)$.
\end{definition}

\subsection{Background on First-Order Logic} \label{subsec:FO-prelims}

In this section, we discuss preliminaries and background on first-order logic (\FO). Almost all of this exposition is taken verbatim from \cite{GroheKS17}. For additional background, we point the reader to \cite{GroheKS17} and references therein.

A \emph{(relational) vocabulary} is a finite set of relation symbols, each with a prescribed arity. Let $\sigma$ be a vocabulary. Then, a $\sigma$-structure $A$ consists of (i) a finite set $V(A)$, called the \emph{universe} or \emph{vertex set} of $A$, and (ii) for each $k$-ary relation symbol $R \in \sigma$, a $k$-ary relation $R(A) \subseteq V(A)^k$. 

Let $A$ be a $\sigma$-structure. For a subset $X \subseteq V(A)$, the \emph{induced substructure} of $A$ with universe $X$ is the $\sigma$-structure $A[X]$ with $V(A[X]) = X$, and $R(A[X]) = R(A) \cap X^k$ for every $k$-ary $R \in \sigma$. For $\sigma' \subseteq \sigma$, the $\sigma'$-restriction of $A$ is the $\sigma'$-structure $A'$ with $V(A') = V(A)$ and $R(A') = R(A)$ for all $R \in \sigma'$. 

For example, a graph $G$ can be thought of as $\{E\}$-structure. This notion is generalized as follows. A \emph{colored graph vocabulary} consists of the binary relation symbol $E$, and possibly finitely many unary relation symbols. A $\sigma$-colored graph is a $\sigma$-structure, whose $\{E\}$-restriction is a simple undirected graph. We call the $\{E\}$-restriction of a $\sigma$-colored graph the \emph{underlying graph} of $G$.

\emph{First-order formulas} of vocabulary $\sigma$ are formed from atomic formulas $x = y$, and $R(x_1, x_2, \ldots, x_k)$, where $R \in \sigma$ is a $k$-ary relation symbol, and $x, y, x_1, \ldots, x_k$ are variables (assuming an infinite supply of variables). Atomic formulas can be combined by Boolean connectives $\neg$ (negation), $\land$ (conjunction), $\lor$ (disjunction), and existential and universal quantification $\exists x, \forall x$, respectively. The set of all first-order formulas of vocabulary $\sigma$ is denoted by $\FO[\sigma]$, and the set of all first order formulas by \FO. The free variables of a formula are the variables that are not in the scope of a quantifier, and we write $\varphi(x_1, \ldots, x_k)$ to indicate that the free variables of the formula $\varphi$ are among $x_1, \ldots, x_k$. A sentence is a formula without free variables.

To define semantics, a satisfaction relation $\models$ is defined inductively as follows. For a $\sigma$-structure $A$, a formula $\varphi(x_1, \ldots, x_k)$, and elements $a_1, \ldots, a_k \in V(A)$, $A \models \varphi(a_1, \ldots, a_k)$ means that $A$ satisfies $\varphi$ if the free variables $x_1, \ldots, x_k$ are interpreted by $a_1, \ldots, a_k$, respectively. If $\varphi(x_1, \ldots, x_k) = R(x_1, \ldots, x_k)$ is an atomic formula, then $A \models \varphi(a_1, \ldots, a_k)$ iff $(a_1, \ldots, a_k) \in R(A)$. The meanings of equality symbol, the Boolean connectives, and the quantifiers are the usual ones.

Whenever a $\sigma$-structure occurs as the input of an algorithm, we assume that it is finite, and is encoded in a suitable way. Similarly, we assume that the formulas $\varphi$ appearing as input are encoded suitably. We denote the length of the encoding of $\varphi$ by $|\varphi|$.

Finally, we will use the following theorem from \cite{GroheKS17}.
\begin{theorem} \label{thm:fo-nowhere}
	For every nowhere dense class $\cal G$, every $\epsilon > 0$, every colored-graph vocabulary $\sigma$, and every first-order formula $\varphi(x) \in \FO[\sigma]$, there is an algorithm that, given a $\sigma$-colored graph $G$ whose underlying graph $G \in \cal G$, computes the set of all $v \in V(G)$ such that $G \models \varphi(v)$ in time $\Oh(n^{1+\epsilon})$.
\end{theorem} 
\section{FPT Algorithm and Kernel Based on Representative Sets} \label{sec:repsets}

In this section we design an algorithm and a kernel for a special case of \chs, using methods based on representative sets~\cite{FominLPS16,Marx09}.  Let $(\cU, \cF, \cB, f: \cB \to \nat, k)$ be an instance of \chs.
The first special case we consider is the following: every element in $\cU$ appears in at most $q$ sets in $\cB$ and any pair of sets in $\cF$ are pairwise disjoint. 

Before this, however, we consider the special case of $q = 1$, i.e., when any pair of sets in $\cB$, as well as that in $\cF$ are disjoint. In this case, we can solve the problem in polynomial time, by reducing it to the problem of finding maximum flow in an auxiliary directed graph, defined as follows. The vertices of the graph are $\cB \uplus \cU \uplus \cF \uplus \{s, t\}$. First, we add arcs (i.e., directed edges) from source $s$ to each $B_j \in \cB$, with capacity $f(B_j)$. Next, for every $u \in B_j$, we add an arc $(B_j, u)$ of capacity $1$. Similarly, for each $u \in F_i$, we add an arc $(u, F_i)$, of capacity $1$. Finally, we add arcs $(F_i, t)$ of capacity $\infty$. It is straightforward to show that there exists a flow of value $k$ in the graph iff there exists a fair hitting set of size $k$. We omit the details.

Note that $q \ge 2$, but the sets in $\cF$ are pairwise disjoint, the problem is \np-hard. In this case, To design both our algorithm and the kernel we first embed the fairness constraints imposed by $\cB$ in a combinatorial object called a \emph{partition matroid}. A partition matroid is a set system $\cM = (E, \cI)$, defined as follows. The ground set $E$ is partitioned into $\ell$ subsets $E_1 \uplus E_2 \uplus \ldots \uplus E_\ell$, such that a set $S \subseteq E$ belongs to the family $\cI$ iff for each $1 \le j \le \ell$, it holds that $|E_j \cap S| \le k_j$, where $k_1, k_2, \ldots, k_\ell$ are non-negative integers. 

It might be observed that the definition of a partition matroid closely resembles the fairness constraints, i.e., for each $B_j $, the hitting set $H$ must satisfy $|H \cap B_i| \le f(B_i)$. However, this idea does not quite work, since the sets $B_j \in \cB$ are not disjoint -- indeed, otherwise we could solve the problem in polynomial time, as discussed earlier. Nevertheless, we can salvage the situation by making $q$ distinct copies of every element $u \in \cU$, and replacing each of the occurrences of $u$ in $q$ distinct $B_j$'s with a unique copy. The resulting set system is a partition matroid that exactly captures the fairness constraints. Correspondingly, in each set of $\cF$, we replace an original element with all of its $q$ copies. Recall that we want to find a hitting set for $\cF$; however, in the new formulation, we must now ensure that if we pick \emph{at least} one copy of element in the solution, we pick \emph{all} of its copies in the solution. 

Thus, our solution is an independent set of $\cM$ that (1) is a hitting set for $\cF$, and (2) picks either $0$ or $q$ copies of every element. To find such a solution in time FPT in $k$ and $q$ (resp.\ to reduce the size of the instance), we use a sophisticated tool developed in parameterized complexity, called \emph{representative sets}. Later, we generalize this idea to the case where very element in $\cU$ appears in at most $q$ sets in $\cB$ \emph{and} at most $d$ sets in $\cF$. In the next section, we formally define the partition matroid, and in the subsequent sections, we apply the toolkit of representative sets to design our FPT algorithm and the kernel.

\subsection{Partition Matroid and Our Solution}\label{subsec:matroid}  In a partition matroid we have a universe $\wU$, partitioned into $\widetilde{U}_1, \cdots ,\widetilde{U}_\ell$, together with positive integers $k_1,\cdots, k_\ell$, and a family of independent sets $\cal I$, such that $X\subseteq \widetilde{U}$ is in $\cI$ if and only if  $|X\cap  \widetilde{U}_i|\leq k_i$, $i\in [\ell]$. 

Let $(\cU,\cB)$ be the given set system such that each element $u \in \cU$ appears in {\em at most}  $q$ sets of $\cB$. For an element   $u \in \cU$, let $q(u)\leq q$ denote the number of sets in $\cB$, that $u$ appears in. Further, for an element $u \in \cU$, let ${\sf copies}(u) = \{ u^1, u^2, \ldots, u^{q(u)} \}$.  We define 
$$ \widetilde{U} =\bigcup_{u\in \cU} {\sf copies}(u) .$$
Next, we need to define a partition of $ \widetilde{U}$. Towards this we use the information about the sets in $\cal B$. We know that each element $u\in B$ appears in $q(u)$ sets and we have made $q(u)$ copies of $u$, thus we use {\em distinct} and {\em unique} copy of $u$ in each sets in $\cB$ in which $u$ appears. This results in $\cwB = \{\wB_i : B_i \in \cB \}$, where $\wB_i$ is the set corresponding to an original set $B_i$, after replacing elements with their copies.  Observe that for every pair of indices $i\neq j$ we have that  $\wB_i \cap \wB_j=\emptyset$ and $\cup_{i} \wB_i =\widetilde{U}$. This immediately gives a partition of  $\widetilde{U}$. Finally, we define $k_i= f(B_i)$. This completes the description of the partition matroid we will be using. We will call this matroid as ${\cal M}=(\widetilde{U}, \cI)$

Given a subset $X \subseteq \widetilde{U}$, we define a set associated with $X$, called  ${\sf projection}(X)$ as follows. The set 
${\sf projection}(X) \subseteq U$, contains an element $u\in U$ if and only if ${\sf copies}(u)\cap X\neq \emptyset$.  Similarly, we define a notion of  embedding. For a set $A \subseteq \cU$, let ${\sf embed}(A) =  \cup_{ u \in A}{\sf copies}(u)$. 
This brings us to the following lemma which relates our problem and finding an independent set in the matroid. 

\begin{lemma}
	\label{lem:mainIset}
	An input 
	$(\cU, \cF, \cB, f: \cB \to \nat, k)$ is a yes-instance if and only if there exists an independent set $X\in \cI$ of the matroid ${\cal M}=(\widetilde{U}, \cI)$ such that (1) $|{\sf projection}(X)|\leq k$, (2) $X = {\sf embed}({\sf projection}(X))$, and (3) ${\sf projection}(X)$ is a hitting set for $\cF$. 
\end{lemma}
\begin{proof}
	In the forward direction, let $S \subseteq \cU$ be a solution to the original problem. Suppose $S = \{s_1, s_2, \ldots, s_k\}$. 
	Consider $X= {\sf embed}(S)$.
	Observe that by definition, $S = {\sf projection}(X)$.  So $|{\sf projection}(X)| = k$ and ${\sf projection}(X)$ is a hitting set for $\cF$. We claim that $X$ is an independent set because if not, there exists a part $\wB_i$ which satisfies that $|X \cap \wB_i| > k_i = f(B_i)$. As $S = {\sf projection}(X)$ and for every $u \in \wB_i$, $B_i$ contains a $v$ such that $u \in {\sf copies}(v)$, this results in $|S \cap B_i| > k_i$ implying $|S \cap B_i| > f(B_i)$ which contradicts that $S$ is a valid solution to the original problem.
	
	In the reverse direction, let $X \in \cI$ be an independent set satisfying both the conditions. Let $S = {\sf projection}(X)$. We claim that, $S$ is a solution for the original instance because if not, there exist a part $B_i$ which satisfies that $|S \cap B_i| > f(B_i) = k_i$. As $X= {\sf embed}(S)$ and for every $u \in B_i$, an unique copy from ${\sf copies}(u)$ is contained in $\wB_i$, which results in $|X \cap \wB_i| > k_i$ which contradicts that $X$ is an independent set. 
\end{proof}

\subsection{Computation of the Desired Independent Set}
In this section we give an algorithm to compute an independent set $X\in \cI$ of the matroid 
${\cal M}=(\widetilde{U}, \cI)$ such that $|{\sf projection}(X)|\leq k$ and ${\sf projection}(X)$ is a hitting set for $\cF$ 
(as given by Lemma~\ref{lem:mainIset}). We will design a dynamic programming algorithm based on representative families to compute the desired independent set.  Towards this, we first give the required definitions. We start with the definition of an {\em $\ell$-representative family}.  

\begin{definition}[{\bf $\ell$-Representative Family}]
	Given a matroid  \mat{} and a family $\cal S$ of subsets of $E$, we say that a subfamily $\widehat{\cal{S}}\subseteq \cal S$ 
	is {\em $\ell$-representative} for $\cal S$ 
	if the following holds: for every set $Y\subseteq  E$ of size at most $\ell$, if there is a set $X \in \cal S$ disjoint from $Y$ with $X\cup Y \in \I$, then there is a set $\whnd{X} \in \whnd{\cal S}$ disjoint from $Y$ with $\whnd{X} \cup  Y \in \I$.  If $\hat{\cal S} \subseteq {\cal S}$ is $\ell$-representative for ${\cal S}$ we write \rep{S}{\ell}. 
\end{definition}

In other words if some independent set in $\cal S$ can be extended to a larger independent set by adding $\ell$ new elements, then there is a set in 
$\widehat{\cal S}$ that can be extended by the same $\ell$ elements. We say that a family  $ \cS = \{S_1,\ldots, S_t\}$ of 
sets is a {\em $p$-family} if each set in $\cal S $ is of size $p$.

\begin{proposition}{\rm \cite[Theorem~$3.8$,~Theorem $1.3$]{FominLPS16,LokshtanovMPS18}}
	\label{thm:repsetlovaszrandomized}
	Let \mat{}   be a partition matroid, $ \cS = \{S_1,\ldots, S_t\}$ be a $p$-family of independent sets. Then there exists \rep{S}{\ell}  of size \bnoml{p+\ell}{p}.  
	Furthermore, given a representation \repmat{M}  of $M$ over a field $ \mathbb{F}$, there is a deterministic algorithm  computing  \rep{S}{\ell}  of size at most   \bnoml{p+\ell}{p} in   \tgemrand \, operations over $ \mathbb{F}$, 
	where $\vert \vert A_M \vert \vert$ denotes the length of $A_M$ in the input. 
\end{proposition}

For the purpose of this article, it is enough to know that partition matroids are 
``representable''~\cite[Proposition $3.5$]{Marx09} and a ``truncation'' of partition matroids are computable in deterministic polynomial time~\cite[Theorem $1.3$]{LokshtanovMPS18}. This results in Proposition~\ref{thm:repsetlovaszrandomized},  which we will use for our algorithm without giving further definitions of representation and truncation~\cite{FominLPS16,LokshtanovMPS18,Marx09}.

Let  $\cF = \{F_1, F_2, \ldots, F_m\}$ be the subsets of $\cU$, $k$ be a positive integer. Since, the sets in $\cF$ are pairwise disjoint, the number of sets in $\cF$ is upper bounded by $k$. 
We call a set $S\subseteq U$, a {\em potential solution}, if for all $j\in [\ell]$, $|S\cap B_j|\leq f(B_j)$.  Let 
\begin{align*}
	{\cS}_i &\coloneqq \{S : S \text{ is a potential solution }, |S|=i \text{ and for all } j\in[i] |S\cap F_j|=1 \}.
\end{align*}
Given ${\cS}_i$, we define ${\cS}_i^{{\sf emb}}$  as $\{{\sf embed}(S)~|~S\in {\cS}_i\}$. Observe that 
${\cS}_i^{{\sf emb}} \subseteq \cI$ and each set has size at most $qi$. Notice that, each set in $\cS$ has size {\em exactly} $i$, but the same can not be said about the sets in ${\cS}_i^{{\sf emb}}$. However, since each element occurs in at most $q$ sets of  $\cB$, we have that each set in ${\cS}_i^{{\sf emb}} $ has size at most $qi$.

Our algorithm checks whether ${\cS}_k$ is non-empty or not. Towards that first observe that  ${\cS}_k$ is non-empty if and only if ${\cS}_k^{{\sf emb}}$ is non-empty. So the testing of non-emptiness of ${\cS}_k$ boils down to checking whether 
${\cS}_k^{{\sf emb}}$ is non-empty or not.  We  test whether ${\cS}_k^{{\sf emb}}$  is non-empty by computing 
$\whnd{\cS}_k^{{\sf emb}} \subseteq_{rep}^0 {\cS}_k^{{\sf emb}}$ and checking whether ${\chS}_k^{{\sf emb}}$ is non-empty. To argue the correctness of the algorithm, first we have the following observation.
\begin{observation} \label{obs:nonempty}
	${\cS}_k^{{\sf emb}} \neq \emptyset$ iff $\whnd{\cS}_k^{{\sf emb}} \neq \emptyset$.
\end{observation}
\begin{proof}
	Since $\whnd{\cS}_k^{{\sf emb}} \subseteq {\cS}_k^{{\sf emb}}$, the reverse direction is immediate. Now we argue the forward direction. Suppose ${\cS}_k^{{\sf emb}}$, then it contains some set $A$. Note that $A$ trivially satisfies $A \cap \emptyset = \emptyset$. Therefore, since $\whnd{\cS}_k^{{\sf emb}} \subseteq_{rep}^0 {\cS}_k^{{\sf emb}}$, there must exist a set $\hat{A} \in \whnd{\cS}_k^{{\sf emb}}$ such that $\hat{A} \cap \emptyset = \emptyset$, i.e., $\whnd{\cS}_k^{{\sf emb}} \neq \emptyset$. 
\end{proof}
Thus, having computed the representative family $\whnd{\cS}_k^{{\sf emb}}$ all we need to do is to check whether it is non-empty. All that remains is an algorithm that computes the representative family  $\whnd{\cS}_k^{{\sf emb}}$.

Let ${\cal Z}$ be a family of sets and $\ell$ be an integer, then ${\cal Z}[\ell]$ is a subset of $\cal Z$ that contains {\em all the sets of} $\cal Z$ of size {\em exactly} $\ell$. We describe a dynamic programming based algorithm. Let 
${\cal D}$ be an array indexed from integers in $\{0,1,\ldots, k\}$. The entry ${\cal D}[i]$ stores the following for all 
$j\in \{i,\ldots, qi\}$, $\whnd{\cS}_i^{{\sf emb}}[j] \subseteq_{rep}^{qk-j} {\cS}_i^{{\sf emb}}[j]$.

We  fill the entries in the matrix $\cal D$ in the increasing order of indexes. 
For $i=0$, ${\cal D}[0]=\emptyset$. Suppose, we have filled all the entries until the index $i$. Then consider the set 
\begin{align*}
	{\cal N}^{i+1} &= \{X'=X\cup {\sf embed}(\{u\}) : X \in {\cal D}[i],~u\in F_{i+1}, {\sf projection}(X') \text{ is a potential solution} \}
\end{align*}
We partition sets in ${\cal N}^{i+1}$  based on sizes. Let ${\cal N}^{i+1}[j]$ denote all the sets in ${\cal N}^{i+1}$ of size $j$. 
\begin{restatable}{claim}{kpathrepset} 
	\label{lem:kpathauxrepset}
	For all $j\in \{i+1,\ldots, q(i+1)\}$, ${\cal N}^{i+1}[j] \subseteq_{rep}^{qk-j} {\cS}_{i+1}^{{\sf emb}}[j]$.
\end{restatable}
\begin{proof}
	Let $S \in {\cS}_{i+1}^{{\sf emb}}[j]$ and  $Y$ be a set of size at most $qk-j$ (which is essentially an independent set of the matroid 
	${\cal M}=(\widetilde{U}, \cI) $) such that  $S\cap Y=\emptyset$ and $S\cup Y \in \cI$. 
	We will show that there exists a set $S' \in {\cal N}^{i+1}[j]$ such that $S' \cap Y=\emptyset$ and $S\cup Y \in \cI$. This will imply  the desired result. Since $S \in {\cS}_{i+1}^{{\sf emb}}[j]$  there exists an element $u\in F_{i+1}$ such that 
	\[S= (S\setminus {\sf embed}(\{u\})) \cup  {\sf embed}(\{u\}).\]
	Let $S_i = (S\setminus {\sf embed}(\{u\}))$. Since, $S$ is an independent set of the matroid ${\cal M}=(\widetilde{U}, \cI) $, we have that $S_i$ is an independent set  of the matroid ${\cal M}=(\widetilde{U}, \cI) $ (hereditary property). Further, 
	$|{\sf projection}(S_i)|= i$ and ${\sf projection}(S_i)$ is a hitting set for $F_1,\ldots,F_i$. This implies that 
	$S_i \in  {\cS}_{i}^{{\sf emb}}$.  Let $Y_i= Y \cup {\sf embed}(\{u\})$. Notice that since $S\cap Y=\emptyset$ and $S\cup Y \in \cI$, we have that $Y_i$ is an independent set and $S_i \cup Y_i = S\cup Y \in \cI$.  Let $|S_i|=j'$. Then, we know that  
	${\cal D}[i]$ contains  $\whnd{\cS}_i^{{\sf emb}}[j'] \subseteq_{rep}^{qk-j'} {\cS}_i^{{\sf emb}}[j']$. This implies that there exists a set $S_i' \in  \whnd{\cS}_i^{{\sf emb}}[j'] $ such that  $S_i' \cup Y_i \in \cI$. This implies that 
	$S_i'\cup {\sf embed}(\{u\})$ is in ${\cal N}^{i+1}$. Further, since $|S|=\sum_{x\in {\sf projection}(S)}  |{\sf embed}(\{x\})|$, we have that $|S_i'\cup {\sf embed}(\{u\})|=|S_i|=j$. This implies that $S_i'\cup {\sf embed}(\{u\})$ is in ${\cal N}^{i+1}[j]$. 
	Thus, we can take $S'=S_i'\cup {\sf embed}(\{u\})$. This completes the proof. 
\end{proof}

We fill the entry for ${\cal D}[i+1]$ as follows. We first compute ${\cal N}^{i+1}$. Observe that 
the sets in ${\cal N}^{i+1}$ have sizes ranging from $i+1$ to $q(i+1)$. Now we apply Proposition~\ref{thm:repsetlovaszrandomized} on each of ${\cal N}^{i+1}[j]$, $j\in\{i+1,\ldots, q(i+1)\}$,  and compute $qk-j$ representative. That is,  we compute 
$\widehat{{\cal N}}^{i+1}[j] \subseteq_{rep}^{qk-j} {\cal N}^{i+1}[j]$. We set 
$${\cal D}[i+1]=\bigcup_{j=i+1}^{q(i+1)}\widehat{{\cal N}}^{i+1}[j].$$

\noindent
Observe that the number of sets in ${\cal D}[i]$ of size $j$ is upper bounded by ${q(k-i)+j \choose j} \leq {qk \choose di}\leq 2^{\cO(qk)}$. Hence, the time taken to compute ${\cal D}[i]$ is upper bounded by $2^{\cO(qk)}n^{\cO(1)}$. Thus, the time taken to compute ${\cal D}[i+1]$ requires at most $qk$ invocations of Proposition~\ref{thm:repsetlovaszrandomized}. 
This itself takes $2^{\cO(qk)}n^{\cO(1)}$ time. This completes the proof, resulting in the following result. 

\repsetsimple*

Theorem~\ref{thm:RepSetSimple} can be generalized to the scenario where every element in $\cU$ appears in at most $q$ sets in $\cB$ and at most $d$ sets in $\cF$.  Observe that if each element appear in at most $d$ sets of $\cF$, then the total number of sets that a subset of size $k$ of $\cU$ can hit is upper bounded by $dk$, else we immediately return that given instance is a NO-instance. Let $S=\{u_1, \ldots, u_k\}$ be a hypothetical solution to our problem. Now, with the help of  $S$, we partition $\cF$ as follows. Let ${\cal F}_i$ denote all sets in $\cF$ that contain $u_i$ and none of 
$\{u_1,\ldots,u_{i-1}\}$. Clearly, ${\cal F}_i$, $i\in[k]$, partitions $\cF$. Now we can design a dynamic programming algorithm similar to the one employed in Theorem~\ref{thm:RepSetSimple}, where in each iteration we grow our representative family by elements that only hit sets in ${\cal F}_i$ and not in ${\cal F}_j$, $j>i$. This will result in the following theorem.

\repsetsimpleplus*

\subsection{A Kernel for a Special Case of \chs using Matroids}
In this section we design a polynomial kernel for the same special case of \chs, that we considered in the last section. 
Let $(\cU, \cF, \cB, f: \cB \to \nat, k)$ be an instance of \chs and assume that every element in $\cU$ appears in at most $q$ sets in $\cB$ and any pair of sets in $\cF$ are pairwise disjoint. To design our kernel we will again use  Lemma~\ref{lem:mainIset} that says that an input 
$(\cU, \cF, \cB, f: \cB \to \nat, k)$ is a yes-instance if and only if there exists an independent set $X\in \cI$ of the matroid ${\cal M}=(\widetilde{U}, \cI)$ such that $|{\sf projection}(X)|\leq k$ and ${\sf projection}(X)$ is a hitting set for $\cF$.

Let  $\cF = \{F_1, F_2, \ldots, F_m\}$ be the subsets of $\cU$, and $k$ be a positive integer. Since, the sets in $\cF$ are pairwise disjoint, the number of sets in $\cF$ is upper bounded by $k$. In particular, we assume that $m=k$. We define 
${\cal F}_i^{{\sf emb}} =\{{\sf embed}(\{u\})~|~u\in F_i\}$. For our kernel we apply the following reduction rules. We start with some simple reduction rules. 

\begin{rrule}
	\label{rr:simple}
	Let $(\cU, \cF, \cB, f: \cB \to \nat, k)$ be an instance of \chs. 
	\begin{itemize}
		\item If there exists an element $u\in \cU$, such that $u$ does not appear in any sets in $\cF$ then delete it from $\cU$ and  all the sets in $\cB$  that it appears in. 
		\item If there exists a set $B\in \cB$ such that $B=\emptyset$, then delete $B$, and take $f$ as the restriction of $f$ on 
		$\cB\setminus \{B\}$. 
		\item If there exists a set $B\in \cB$ such that $f(B)=0$, then we do as follows: 
		$\cU := \cU \setminus \{ B\}$; delete all the elements of $B$ from all the sets in $\cB$ and $\cF$ that it appears in. If some set in $\cal F$ becomes empty then return a trivial \textsc{No}-instance. Else, take $f$ as the restriction of $f$ on 
		$\cB\setminus \{B\}$ and keep the integer $k$ unchanged. 
	\end{itemize}
\end{rrule} 

Soundness of Reduction Rule~\ref{rr:simple} is obvious and hence omitted. The next reduction rule is the main engine of our kernel. 
\begin{rrule}
	\label{rr:elementReduction}
	Let $(\cU, \cF, \cB, f: \cB \to \nat, k)$ be an instance of \chs. If there exists a pair of integers $i\in[k]$ and $j\in[q]$ such that 
	$|{\cal F}_i^{{\sf emb}}[j]| > {kq \choose j}$, then do as follows. Compute $\widehat{\cal F}_i^{{\sf emb}}[j] \subseteq_{rep}^{qk-j}  {\cal F}_i^{{\sf emb}}[j]$. Let $F\in {\cal F}_i^{{\sf emb}}[j]$ that do not appear in $\widehat{\cal F}_i^{{\sf emb}}[j] $. Then obtain a reduced instance as follows. 
	\begin{itemize}
		\item $\cU := \cU \setminus  {\sf projection}(F)$
		\item Delete ${\sf projection}(F)$ from all the sets in $\cB$ and $\cF$ that it appears in. 
		\item The function $f$ and $k$ remains the same. 
	\end{itemize}
	
\end{rrule}

\begin{restatable}{lemma}{elementreduction} \label{lem:element-reduction}
	Reduction Rule~\ref{rr:elementReduction} is sound. 
\end{restatable}
\begin{proof}
	Let  $(\cU, \cF, \cB, f: \cB \to \nat, k)$ be an instance of \chs, and let $(\cU', \cF', \cB', f: \cB \to \nat, k)$ 
	be the reduced instance, after an application of Reduction Rule~\ref{rr:elementReduction}. It is easy to see that a solution to the reduced instance can directly be lifted to the input instance. Thus, we focus on forward direction. 
	
	In the forward direction, let $S$ be a solution to $(\cU, \cF, \cB, f: \cB \to \nat, k)$. Then, by Lemma~\ref{lem:mainIset}, it implies that  ${\sf embed}(S) \in \cI$ (of the matroid ${\cal M}=(\widetilde{U}, \cI)$) and  $|S|\leq k$ and $S={\sf projection}({\sf embed}(S))$ is a hitting set for $\cF$. Let $u = {\sf projection}(F)$. Then, we have that $\cU := \cU \setminus \{ u\}$. 
	If $u \notin S$, then $S$ is also the solution to $(\cU', \cF', \cB', f: \cB' \to \nat, k)$. So we assume that $u\in S$.  
	
	
	Observe that $F={\sf embed}(\{u\})$, $|F|=j$, and $u$ belongs to $F_i$. Further, since every set in $\cal F$ are pairwise disjoint we have that the only job of $u$ is to hit the set $F_i$. Consider, $Y= {\sf embed}(S) \setminus {\sf embed}(\{u\})$. Since, ${\sf embed}(S) \in \cI$, we have that $Y\in \cI$ (hereditary property of the matroid), and the size of 
	$Y$ is upper bounded by $qk-j$. The last assertion follows from the fact that for any element $v\in \cU$,  the size of ${\sf embed}(\{v\})$ is upper bounded by $q$ and $|{\sf embed}(S)|=\sum_{x\in S}  |{\sf embed}(\{x\})| \leq qk$. This implies that there exists $F' \in \widehat{\cal F}_i^{{\sf emb}}[j] \subseteq_{rep}^{qk-j}  {\cal F}_i^{{\sf emb}}[j]$ such that $Y\cup \{F' \} \in \cI$.  
	Since, $|{\sf projection}({\sf embed}(S))|\leq k$, we have that $|{\sf projection}(Y)|\leq k-1$. Thus, $|{\sf projection}(Y\cup \{F' \} )|\leq k$. Now we need to show that ${\sf projection}(Y\cup \{F' \}) $ is a hitting set for $\cF$. This follows from the fact that 
	$u' = {\sf projection}(F') \in F_i$. In other words, we have shown that  $S' =  S \setminus \{u\} \cup \{u'\}$ is a desired hitting set for $\cF$. This concludes the proof.
\end{proof}

Finally, we get the following kernel.

\begin{restatable}{theorem}{kernelthm}
	\label{thm:kernel}
	Let $(\cU, \cF, \cB, f: \cB \to \nat, k)$ be an instance of \chs such that every element in $\cU$ appears in at most $q$ sets in $\cB$ and any pair of sets in $\cF$ are pairwise disjoint. Then, \chs admits a kernel of size 
	$\cO(kq^2 {kq \choose q} \log k)$. 
\end{restatable}
\begin{proof}
	For our algorithm we apply Reduction Rules~\ref{rr:simple} and ~\ref{rr:elementReduction} exhaustively. If any application of these rules return that the input is a \textsc{No}-instance, we return the same. The correctness of the algorithm follows from the correctness of Reduction Rules~\ref{rr:simple} and~\ref{rr:elementReduction}. Further it is clear that the algorithm runs in polynomial time.  What remains to show is that the reduced instance is upper bounded by the claimed function.
	
	For convenience we assume that the reduced instance is also denoted by  $(\cU, \cF, \cB, f: \cB \to \nat, k)$. Since, 
	Reduction Rule~\ref{rr:elementReduction} is not applicable we have that each set $F_i$, $i\in[k]$, is upper bounded by 
	$\sum_{j=1}^q  {kq \choose j} \leq q  {kq \choose q}$. This implies that $|\cU| \leq kq {kq \choose q}$. Further, since 
	every element in $\cU$ appears in at most $q$ sets in $\cB$, we have that the number of non-empty sets in $\cB$ is upper bounded by $q |\cU| \leq kq^2 {kq \choose q}$. Since, Reduction Rule~\ref{rr:simple} is not applicable we have that there are no empty-sets and hence $|\cB| \leq kq^2 {kq \choose q} $. Further, to represent the function $f$ we need at most 
	$\cO(|\cB| \log k)  \leq \cO(kq^2 {kq \choose q} \log k)$ bits. This completes the proof. 
\end{proof}

\section{Reduction from $K_{d, d}$-free $G_{\cU, \cF}$ to Bounded Frequency in $\cF$} \label{sec:kddfree-freq}

We give an alternate, self-contained proof of this theorem in \Cref{sec:kddfree-freq}. This reduction is used as the first step in some of our results, such as \Cref{sec:nowhere}. We note that a part of the kernelization procedure in \Cref{sec:kdd-kernel} can also be used to prove the following theorem. 

\begin{restatable}{theorem}{kdd}\label{thm:kdd-free}
	Let $\cI = (\cU, \cF, \cB, f, k)$ be an instance of \chs, where $G_{\cU, \cF}$ is $K_{d, d}$ free, and $G_{\cU, \cB}$ satisfies a hereditary property $\Pi$.
	
	Suppose there exists an algorithm $\mathcal{A}$ that, given an instance $\cI'=(\cU',\cF',\cB',f',k')$ such that (1) every element of $\cU'$ appears in at most $\Delta$ sets of $\cF'$, and (2) $G_{\cU', \cB'}$ satisfies $\Pi$, decides in time $h(k',\Delta)\cdot |\cI'|^{\Oh(1)}$, whether $\cI'$ is a yes-instance of \chs. Then, there exists an algorithm that runs in time $g(k,d)\cdot |\cI|^{\Oh(1)}$ to decide whether $\cI$ is a yes-instance of \chs. 
\end{restatable}
\begin{proof}
	Let $G=G_{\cU, \cF}$ be an incidence graph corresponding to the set system $(\cU, \cF)$ in instance $\cI$. The proof of the theorem relies on the following modification of a classical lemma from parameterized algorithms. The lemma guarantees that, as long as there exists an element that is contained in a large number of sets in $\cF$, we can branch on a set of at most $d-1$ elements. 
	\begin{restatable}{lemma}{lemkdd} \label{lem:highdegreeintersectssolution}
		If there exists an element $u \in \cU$ such that $deg_G(u)\geq dk^{d-1}$, then in polynomial time we can find a set $X \subseteq \cU$ of size at most $d-1$ that intersects every fair hitting set for $\cF$ of size at most $k$.
	\end{restatable}
	\begin{proof}
		Suppose there exists $u \in \cU$ such that $deg_{G}(u)\geq dk^{d-1}$.
		Let $X =\{u_1,\ldots,u_\ell\}\subseteq \cU$ be a maximal set such that for all $p \leq \ell$ we have that $\bigcap_{i=1}^{p}N_G(u_i) \geq dk^{d-p}$.
		Such a set $X$ can be found in polynomial time by greedily selecting the vertices.
		Observe that $\ell \leq d-1$, otherwise	$\bigcap_{i=1}^{\ell}N_G(u_i) \geq dk^{d-\ell} \geq 
		d$, , it would imply an existence of $K_{d,d}$ in $(\cU, \cF)$.
		
		Let $\cF_{X}=\bigcap_{i=1}^{\ell}N_G(u_i)$, and note that $|\cF_{X}| \geq dk^{d-\ell}$.
		For every vertex $u' \in \cU \setminus X$, we have
		$|N(u') \cap \cF_{X}| < dk^{d-\ell-1}$. 
		We next claim that every fair hitting set $H \subseteq \cU$ for $\cF$ of size at most $k$ satisfies $H \cap X \neq \emptyset$.
		Fix, a hitting set $H'$ for $\cF$ of size $k$ and suppose $H' \cap X =\emptyset$.
		Then the number of sets of $\cF_X$ that are hit by $k$ elements of $H' \subseteq \cU \setminus X$ is less than $k \cdot dk^{d-\ell-1}= dk^{d-\ell} \le |\cF_X|$. Thus, there exists a set in $\cF_X \subseteq \cF$ not hit by $H'$, which contradicts the assumption that $H'$ hits $\cF$.
		Therefore, $|H' \cap X|\neq \emptyset$.
	\end{proof}

	Suppose there exists an element with at least $d k^{d-1}$ neighbors in $\cF$. Then, we use the algorithm from \Cref{lem:highdegreeintersectssolution} to obtain a set  $X=\{u_1,\ldots u_\ell\}$, where $\ell \le d-1$. Now, we branch on each element $u \in X$. More specifically, we make a recursive corresponding to a new instance defined w.r.t. $u$, as follows. $\tilde{\cI}=(\tilde{\cU}, \tilde{\cF}, \tilde{\cB}, \tilde{k},\tilde{f})$, 
	where $\widetilde{\cU} = \cU \setminus \{u\}$, $\widetilde{\cF} = \cF - \{u\}$, $\widetilde{\cB} = \cB - \{u\}$, $\widetilde{k} = k-1$, and for every $B_j \in \cB$, we set
	$\tilde{f}(B_j) = \begin{cases}
		f(B_j) & \text{ if $u \not\in B_j$}
		\\f(B_j) - 1 & \text{ if $u \in B_j$}
	\end{cases}$.\\
	
	Else, there is no element with at least $dk^{d-1}$ neighbors in $\cF$.
	Thus, every element in $\cU$ appears in  at most $dk^{d-1}$ sets in $\cF$.
	We set $\Delta=dk^{d-1}$, then  use algorithm $\mathcal{A}$ to solve the instance $\cI$ in $h(k,\Delta)\cdot|\cI|^{\mathcal{O}(1)}$ time.

	Now we analyze the running time. At every recursive call we decrease $k$ by 1, and thus the height of the search tree does not exceed $k$.
	At every step we branch on at most $d-1$ subproblems.
	Hence, the number of leaves in the search tree is at most $(d-1)^{k}$. 
	At each leaf node either $k \leq 0$, in which case we solve instance $\cI$ in polynomial time.
	Otherwise, a leaf node corresponds to the case when we cannot branch, since \Cref{lem:highdegreeintersectssolution} cannot be applied. In this case, each element of $\cU$ is contained in at most $\Delta \coloneqq dk^{d-1}$ sets of $\cF$. In this case, we use algorithm $\mathcal{A}$, which solves instance $\cI'$ in $h({k'},\Delta)\cdot|{\cI'}|^{\mathcal{O}(1)}=h'(k,d)\cdot|\cI|^{\mathcal{O}(1)}$ time, for some function $g$.
	Thus, the total time taken by the algorithm is $(d-1)^k \cdot h'(k,d)\cdot|\cI|^{\mathcal{O}(1)}= g(k,d)\cdot|\cI|^{\mathcal{O}(1)}$, as desired.
\end{proof}

\section{Polynomial Kernel for $K_{d,d}$-free $G_{\cU, \cF}$ and Bounded Frequency in $\cB$} \label{sec:kdd-kernel}
Consider an input $(\cU, \cF, \cB, f: \cB \to \nat, k)$ of \chs. In this section we design a polynomial kernel for \chs problem when $G_{\cU, \cF}$ is $K_{d,d}$-free and frequency of each element in $\cB$ is at most $q$. We fix $d$ and $q$ for the rest of the section. Without loss of generality we assume that $d\geq 2$, $k\geq 2$. We also assume that we do not have multisets in $\cF$ and $\cB$.  The kernelization algorithm consists of two phases: In phase one we apply some reduction rules to bound the size of $|\cF|$. In phase two we use partition matroid ${\cal M}=(\widetilde{U}, \cI)$ defined in Section~\ref{subsec:matroid} using $\cB$ to design a reduction rule to bound the number of elements. Next, we state our reduction rules that the kernelization algorithm applies in the order in which they are stated. For a subset $\cB'\subseteq \cB$, by $f'=f_{|\cB'}$, we denote the function $f$ restricted to $\cB'$. We begin with the following simple reduction rule that removes empty sets from $\cB$.

\begin{rrule}\label{rr:emptyB}
	If there exists a set $B\in \cB$ such that $B=\emptyset$, then delete $B$, and take $f$ as the restriction of $f$ on $\cB\setminus \{B\}$. Return the instance $\cI'=(\cU,\cF,\cB',f',k)$, where $\cB'=\cB\setminus \{B\}$ and $f'=f_{|\cB'}$.
\end{rrule}

Next reduction rule removes sets in $\cB$ that do not interest with any set in $\cF$.

\begin{rrule}\label{rr:mindegreeF}
	If there exists an element $u\in \cU$ such that $u$ is not contained in any set in $\cF$, that is for every $F\in \cF$, $u\notin F$, then remove $u$ from every set in $B$. Return the instance $\cI'=(\cU',\cF,\cB',f,k)$, where $\cU'=\cU\setminus \{u\}$, $\cB'=\cB-\{u\}$.
\end{rrule}

The next reduction rule removes sets in $\cB$ that have $f$ value less than 1.
\begin{rrule}\label{rr:fatleast1}
	If there exists a set $B\in \cB$ such that $f(B)\leq 0$, if there exists a set $F\in \cF$ such that $F\subseteq B$, then return that $(\cU, \cF, \cB, f, k)$ is a no-instance of \chs, otherwise from every set $F\in \cF$ delete elements that are contained in $B$. Return the instance $\cI'=(\cU',\cF',\cB',f',k)$, where $\cU'=\cU\setminus B$, $\cF'=\cF-B$, $\cB'=\cB\setminus \{B\}$, $f'=f_{|\cB'}$.
\end{rrule}

We will use the following variant of Lemma~\ref{lem:highdegreeintersectssolution} to design our next reduction rule which will bound the degree of a vertex in $U$ in the graph $G_{U,\cF}$. 

\begin{claim}\label{claim:highdegreeintersectssolution}
	If there exists an element $u \in \cU$ such that $|\cF(u)|\geq dk^{d-1}$, then in polynomial time we can find a inclusion-maximal set $X \subseteq \cU$ of size at most $d-1$ such that $|\bigcap_{v\in X} \cF(v)| \geq dk^{d-|X|}$ and for every fair hitting set $S$ for $\cF$ of size at most $k$, $S\cap X\neq \emptyset$.
\end{claim}

Observe that in the above claim when $|X|>1$, $|\bigcap_{v\in X} \cF(v)| \geq 2$. The proof of Lemma~\ref{lem:highdegreeintersectssolution} also holds for the above claim.

\begin{rrule}\label{rr:degreereduction}
	If there exists an element $u \in \cU$ such that $\cF(u)\geq dk^{d-1}$, then let $X$ be the output of the algorithm in Claim~\ref{claim:highdegreeintersectssolution}. Consider the following cases:
	\begin{enumerate}
		\item If $|X|=1$, and let $X=\{u\}$, then proceed as follows; Delete all the sets containing $u$ from $\cF$. Let $\cF'=\{F|u\notin F,F\in \cF\}$.  Next,delete $u$ from every set $B\in \cB$ and reduce $f(B)$ by 1, let $\cB'=\cB-\{u\}$. Finally reduce $k$ by 1. Return the instance $\cI'=(\cU',\cF',\cB',f,k')$, where $\cU'=\cU\setminus \{u\}$, $f'=f(B)-1$, if $u\in B$, otherwise $f'=f(B)$ and $k'=k-1$.
		\item If $|X|>1$, then we delete all the sets from $\cF$ that contain the set $X$ and add the set $X$ to $\cF$. Return the instance $\cI'=(\cU',\cF',\cB,f,k)$, where $\cF'=\{F|X\nsubseteq F,F\in \cF\}\cup \{X\}$, $\cU'=\bigcup_{F\in \cF'}F$.
	\end{enumerate}
\end{rrule}
\begin{lemma}
	Reduction Rule \ref{rr:degreereduction} is correct, and either reduces the size of $\cU$ or that of $\cF$.
\end{lemma}

\begin{proof}
	Note that $G_{\cU',cF'}$ remains $K_{d,d}$-free even after adding the set $X$ to $\cF'$. The proof of case when $|X|=1$ follows from the Claim \ref{claim:highdegreeintersectssolution}. Next we prove lemma for case when $|X|>1$. To show that the lemma holds we will show that $(\cU, \cF, \cB, f, k)$ is a yes-instance of \chs if and only if $(\cU',\cF',\cB,f,k)$ is a yes instance of \chs. In the forward direction suppose that $(\cU, \cF, \cB, f, k)$ is a yes-instance of \chs and let $S$ be its solution. By Claim \ref{claim:highdegreeintersectssolution}, $S\cap X\neq \emptyset$, therefore $S$ is also a hitting set of $\cF'$ and hence a fair hitting set of $(\cU',\cF',\cB,f,k)$. In the backward direction suppose that $(\cU',\cF',\cB,f,k)$ is a yes instance of \chs. and let $S'$ be its solution. We claim that $S'$ is also a hitting set of $\cF$. Suppose that for a contradiction there exists $F\in \cF$, that is not hit by $S'$. Then $F\notin \cF'$, however by our construction of $\cF'$, $X\subseteq F$. As $X\in \cF'$, $X\cap S'\neq \emptyset$ and $S'$ must hit $F$, a contradiction. Hence $S'$ is a fair hitting set of $(\cU, \cF, \cB, f, k)$. This concludes the proof.
\end{proof}

Observe that the Reduction Rule \ref{rr:degreereduction} is not applicable in polynomial time and only polynomial many times as in the first case when $|X|=1$ since Reduction Rule \ref{rr:mindegreeF} is not applicable the sets in $\cF$ containing $X$ must be at least 1, and in second case when $|X|>1$, by Claim \ref{claim:highdegreeintersectssolution}, the sets in $\cF$ containing $X$ must be at least $2$.

Observe that when Reduction Rule \ref{rr:degreereduction} is no longer applicable then each element is contained in at most $dk^{d-1}$ sets in $\cF$.

\begin{rrule} \label{rrule:f-bounded}
	If Reduction Rules \ref{rr:mindegreeF}-\ref{rr:degreereduction} are no longer applicable and $|\cF|>dk^d$, then return that $(\cU, \cF, \cB, f, k)$ is a no-instance of \chs.
\end{rrule}
Note that if \Cref{rrule:f-bounded} is not applicable, then the size of $|\cF|$ is bounded by $dk^{d}$.
\begin{remark} \label{rem:preproc}
	Note that in \Cref{rr:emptyB} to \Cref{rrule:f-bounded}, we did not use any property of the $(\cU, \cB)$ set system, except that we may delete an element of $\cU$ or a set from $\cB$. Thus, in the resulting instance, the $G_{\cU, \cB}$ satisfies the property $\Pi$. Therefore, we can use \Cref{rr:emptyB} to \Cref{rrule:f-bounded} as an alternate way of obtaining a proof of \Cref{thm:kdd-free}.
\end{remark}

\begin{remark}
	As observed before, at this point the size of $|\cF|$ is upper bounded by $dk^d$. Now, we use an approach similar to \Cref{thm:RepSetSimplePlus} to obtain an FPT algorithm for \chs that runs in time $d^{k}k^{kd} \cdot 2^{\Oh(kq)} \cdot |\cI|^{\Oh(1)}$.
\end{remark}

\noindent\textbf{Back to kernelization.}
Let $A_{\geq d}$ be the set of all the elements with degree at least $d$ in $G_{\cU,\cF}$.  As $G_{\cU,\cF}$ is $K_{d,d}$-free we obtain the following observation. 

\begin{observation}\label{obs:largedegree}
	$|A_{\geq d}|\leq {{|\cF|}\choose{d}}\cdot d$.
\end{observation}

For a set $\cY\subseteq \cF$ of size at most $d-1$, we denote the set ${\sf ExactNbr}(\cY)$ to be the set of all the elements $u\in \cU$ such that $\cF(u)=\cY$, that is $u$ is only contained in all the sets in $\cY$.  
Next, we state a reduction rule that bounds elements in ${\sf ExactNbr}(\cY)$. For this purpose we use partition matroid ${\cal M}=(\widetilde{U}, \cI)$ defined in Section~\ref{subsec:matroid}. Recall that for constructing $\cM$, we add ${\sf copies}(u) = \{ u^1, u^2, \ldots, u^{deg_{G_{\cU,\cB}}(u)} \}$ for each element $u\in \cU$ to $\widetilde{U}$. For ease of proof we add dummy copies of every element to make the frequency of every element exactly $q$ in $\cB$, that is we let ${\sf copies}(u) = \{ u^1, u^2, \ldots, u^q \}$ for every element $u\in \cU$. We add a new set $B^*$ to $\cB$ and for every element $u$ with frequency $q'<q$ in $\cB$, we add  $\{ u^{q'+1}, \ldots, u^q \}$ to $B^*$. We let $f(B^*)=kq$. Now we define the partition matroid ${\cal M}=(\widetilde{U}, \cI)$ as defined in Section~\ref{subsec:matroid}. Next, we define a family $\cP_{\cY}=\{{\sf copies}(u)|u\in {\sf ExactNbr}(\cY)\}$ and apply the for reduction rule for every set $\cY\subseteq \cF$ of size at most $d-1$.

\begin{rrule}
	\label{rr:elementRedUsingB}
	Let $(\cU, \cF, \cB, f: \cB \to \nat, k)$ be an instance of \chs. If there exists a set $\cY\subseteq \cF$ such that $|\cY|\leq d-1$ and
	${\sf ExacrNbr}(\cY) > {kq \choose j}$, then do as follows. Compute $\widehat{\cP}_{\cY} \subseteq_{rep}^{kq-q}  {\cP}_{\cY}$. Let $u\in {\sf ExacrNbr}(\cY)$ such that  ${\sf copies}(u)\in {\cP}_{\cY}$ and ${\sf copies}(u)\notin \widehat{\cP}_{\cY}$. Then delete $u$ from every set in $\cF$ and every set in $\cB$. Let $\cU'=\cU\setminus \{u\}$, $\cF'=\{F\setminus \{u\}|F\in \cF\}$, $\cB'=\{B\setminus \{u\}|B\in \cB\}$, and $f'(B\setminus \{u\})=f(B)$ for each $B\in \cB$. Return the instance $\cI'=(\cU',\cF',\cB',f,k)$. 
\end{rrule}

\begin{lemma}
	Reduction Rule \ref{rr:elementRedUsingB} is correct.
\end{lemma}
\begin{proof}
	To show that the lemma holds we will show that $(\cU, \cF, \cB, f, k)$ is a yes-instance of \chs if and only if $(\cU',\cF',\cB',f,k)$ is a yes instance of \chs. Observe that the backward direction holds trivially. To prove the forward direction suppose that $(\cU, \cF, \cB, f, k)$ is a yes-instance of \chs and let $S\subseteq \cU$ be its solution of size at most $k$. Observe that if $u\notin S$, then $S$ is also a solution to $(\cU',\cF',\cB',f,k)$. Suppose that $u\in S$. Let $\cS=\{{\sf copies}(v)|v\in S\}$. By Lemma~\ref{lem:mainIset} $\cS$ is an independent set of the matroid ${\cal M}=(\widetilde{U}, \cI)$. Let $\cS'=\cS\setminus \{{\sf copies}(u)\}$. Notice that $|\cS'|\leq kq-q$, as we have frequency exactly $q$ for each element of $\cU$ in $\cB$. Therefore as  $\widehat{\cP}_{\cY}$, is a ${kq-q}$ representative  family of  ${\cP}_{\cY}$ and $ {\sf copies}(u)\notin \widehat{\cP}_{\cY}$ there must exists an element $v\in {\sf ExacrNbr}(\cY)$ such that $ {\sf copies}(v)\in \widehat{\cP}_{\cY}$ and $\cS'\cup \{{\sf copies}(v)\}$ must be an independent set of the matroid ${\cal M}=(\widetilde{U}, \cI)$. Let $S'={\sf projection}(\cS')$, by our construction of matroid ${\cal M}$, $S'$ satisfies $f$. Also observe that $S'$ is a hitting set of $\cF$ since $S\setminus \{u\}=S'\setminus \{v\}$ and $\cF(u)=\cF(v)=\cY$.
\end{proof}

\begin{observation}\label{obs:smalldegree}
	When Reduction Rule \ref{rr:elementRedUsingB} is no longer applicable, then for every $\cY \subseteq \cF$ size of the set ${\sf ExactNbr}(\cY)$ is bounded by ${kq \choose q}$. 
\end{observation}

By the above observation the number of elements with degree at most $d-1$ in $G_{\cU,\cF}$ is bounded by $d\cdot {{|\cF|}\choose{d}}\cdot {kq \choose q} $. Therefore by observation \ref{obs:largedegree} the number of elements in $\cU$ is bounded by $\Oh(k^{\Oh(d^2+q)}d^{d}q^q)$. Observe that each of our reduction rules can be applied in polynomial time and only polynomial times. Thus we obtain the following theorem. 

\thmKdd*
\section{Parameterization by $k+d$ when $G_{\cU, \cB}$ is Nowhere Dense and $G_{\cU, \cF}$ is $K_{d,d}$-free} \label{sec:nowhere}
In this section, we give a formal proof of \Cref{thm:nowhere-dense}. For convenience, we restate the theorem.
\nowhere*

We prove \Cref{thm:nowhere-dense} by reducing the problem to \FO model checking on a colored graph, whose $\{E\}$-restriction is equal to $G$. 

We define the following notation. For a non-empty set $A$, and an integer $0 \le b \le |A|$, we use $\binom{A}{\le b}$ to denote the family of all subsets of $A$ of size at most $b$. For each subset $\cF' \subseteq \cF$, we define $\cU_{\cF'} \coloneqq \{ u \in \cU: \cF(u) = \cF' \}$, that is, $\cU_{\cF'}$ is the subset of elements $u$ such that $\cF'$ is exactly the subset of $\cF$ containing $u$ (alternatively, $\cF'$ is the subset hit by $u$). Note that $\cP_U \coloneqq \{\cU_{\cF'} : \cF' \subseteq \cF \text{ and } \cU_{\cF'} \neq \emptyset \}$ is a partition of the universe $\cU$ into at most $2^m$ parts, where $m = |\cF|$. Similarly, for $1 \le t \le k$, we define $\cB_{t} \coloneqq \{B_j \in \cB: f(B_j) = t\}$. Again, observe that $ \cP_B \coloneqq \{\cB_t: 1 \le t \le k \text{ and } \cB_t \neq \emptyset \}$ is a partition of $\cB$ into at most $k$ parts

Let $\sigma' = \{E(G)\}$, and recall that recall that the graph $G$ can be thought of as an $\sigma'$-structure. Now, we define $\sigma = \sigma' \cup \cP_U \cup \cP_B$, where we think of each $\cU_{\cF'}$, and $\cB_{t}$ as a unary relation on $V(G)$. Note that the number of relations in $\sigma$ is at most $1 + k + 2^m$. Let $G'$ be the resulting $\sigma$-structure, and note that its underlying graph $G$ belongs to a nowhere dense class $\cG$. 

We need the following definition.
\begin{definition}[Valid Pair] \label{def:valid-pair}
	We say that a pair $(k', \bbF)$ is a \emph{valid pair} if it satisfies the following properties. 
	\begin{enumerate}
		\item $k' \in [k]$, and
		\item $\bbF = (\cF_1, \cF_2, \ldots, \cF_{k'})$, where each $\cF_i$ is a non-empty subset of $\cF$, and $\bigcup_{i = 1}^{k'} \cF_i = \cF$.
	\end{enumerate}
\end{definition} 
We observe that the number of valid pairs is upper bounded by $k \cdot 2^{\Oh(km)}$.

Next, corresponding to each valid pair $(k', \bbF)$, we define a formula in $\FO[\sigma]$, as follows:
\begin{definition}[Formula corresponding to a valid pair] \label{def:formula-pair}
	Fix a valid pair $(k', \bbF)$, where $\bbF = (\cF_1, \cF_2, \ldots, \cF_{k'})$. We define $\varphi_{k', \bbF} \in \FO[\sigma]$ as follows:
	\begin{align*}
		\varphi_{k', \bbF} \coloneqq \exists u_1, u_2, \ldots, u_{k'}
		\\ \lr{ \textup{\texttt{enforce-}}\bbF(U) \land \lr{ \bigwedge_{t \in [k]: \cB_t \neq \emptyset } \psi_{t}(U) }}
	\end{align*}
	where, we use the shorthand $U = (u_1, u_2, \ldots, u_{k'})$, and
	\begin{itemize}
		\item $\textup{\texttt{enforce-}}\bbF(U) \coloneqq \bigwedge_{i = 1}^{k'} (u_i \in \cU_{\cF_i})$, 
		\item For $1 \le t \le k$, 
		\begin{align*}
			&\psi_t(U) \coloneqq \forall B_t,
			\\&\qquad(B_t \in \cB_t)\ \implies\ \bigvee_{Q \in \binom{[k']}{\le t}} \textup{\texttt{exact-nbr}}_{Q}(B_t, U) 
		\end{align*}
		\item For $1 \le t \le k$, and $Q \in \binom{[k']}{\le t}$,
		\begin{align*}
			\textup{\texttt{exact-nbr}}_{Q}(B_t, U) \coloneqq &\lr{ \bigwedge_{j \in Q} \{B_t, u_j\} \in E(G) } \land
			\\& \lr{ \bigwedge_{j' \not\in Q}  \{B_t, u_{j'} \} \not\in E(G) } 
		\end{align*}
	\end{itemize}
\end{definition}
\begin{observation}\label{obs:formula-size}
	For any valid pair $(k', \bbF)$, it holds that $|\varphi_{k', \bbF}| \le h'(k, m)$ for some function $h(\cdot, \cdot)$, assuming a suitable encoding.
\end{observation}
\begin{proof}
	For every $1 \le t \le k$, and $Q \in \binom{[k]}{\le t}$, the formula $\textup{\texttt{exact-nbr}}_t, Q(B_t, U)$ is a conjunction of $2^m$ atomic formulas of constant size. Then, for $1 \le t \le k$, the formula $\psi_t(U)$ is a conjunction of an atomic formula of a constant size, and the disjunction of $\sum_{i = 0}^t \binom{k}{t} \le 2^{k}$ formulas of the type $\textup{\texttt{exact-nbr}}_t, Q(B_t, U)$, and the size of each such formula is bounded by $2^{\Oh(m)}$. Therefore, the size of each $\psi_t(U)$ is upper bounded by $2^{\Oh(km)}$. The formula $\textup{\texttt{enforce-}}\bbF(U)$ is of size $\Oh(k')$. Finally, the formula $\varphi_{k', \bbF}$ is obtained by taking conjunction of $\textup{\texttt{enforce-}}\bbF(U)$, along with $\psi_{1}(U), \psi_{2}(U), \ldots, \psi_{k}(U)$. 
	Therefore, the size of $\varphi_{k', \bbF}$ is upper bounded by $2^{\Oh(km)}$.
\end{proof}

Now, we are ready to prove the main technical lemma of this section.
\begin{lemma} \label{lem:formula-equiv}
	$\cI$ is a yes-instance of \chs iff for some valid pair $(k', \bbF')$, it holds that $V(G) \models \varphi_{k', \bbF}$.
\end{lemma}
\begin{proof}
	\textbf{Forward direction.} Let $S \subseteq \cU$ be a fair hitting set of $\cF$ of size $1 \le k' \le k$. Let us arbitrarily number the vertices of $S$ as $v_1, v_2, \ldots, v_{k'}$. For each $1 \le i \le k'$, let $\cF_i = \cF(u_i)$, i.e., $u_i \in \cU_{\cF_i)}$. Let $\bbF = (\cF_1, \cF_2, \ldots, \cF_{k'})$. We claim that $V(G) \models \varphi_{k', \bbF}$, for the setting $u_i = v_i$ of the existentially quantified variables.
	
	First, it is easy to see that $\textup{\texttt{enforce-}}\bbF(U)$ is satisfied since every $v_i$ belongs to $\cU_{\cF_i)}$ by construction. Next, we claim that every $\psi_t$, such that $\cB_t \neq \emptyset$ is satisfied. Fix one such $t$. Note that for the formula $\psi_t$, it suffices to focus on $B_t \in \cB_t$, i.e., for the $B_t$'s such that $f(B_t) = t$. In other words, it suffices to show that for every $B_t \in \cB_t$, the formula $\textup{\texttt{exact-nbr}}_{Q, t}(B_t, U)$ is satisfied for some $Q \in \binom{[k]}{\le t}$. To this end, let $Q = \{ j \in [k'] : v_j \in B_t \}$. Note that since $S$ is a fair hitting set, $|Q| \le f(B_t) = t$. It is easy to verify that the formula $\textup{\texttt{exact-nbrhood}}_{Q, t}(B_t, U)$ is satisfied, since it checks whether the intersection of $S$ with $B_t$ is exactly the set corresponding to $Q$. Thus, $V(G) \models \varphi_{k', \bbF}$.
	
	\textbf{Reverse direction.} Suppose $V(G) \models \varphi_{k', \bbF}$ for some valid pair $(k', \bbF)$, with $\bbF = (\cF_1, \cF_2, \ldots, \cF_{k'})$. For $1 \le i \le k'$, let $v_i$ be equal to the existentially quantified vertex $u_i$ in a model of $\varphi_{k', \bbF}$. Let $S = \{v_1, v_2, \ldots, v_{k'}\}$. We claim that $S$ is a fair hitting set for $\cF$. Note that for $1 \le i \le k'$, $v_i \in \cF_{k'}$. Then, from the definition of a valid pair, it follows that $\bigcup_{i = 1}^{k'} \cF_{i} = \cF$, which implies that $S$ hits every set in $\cF$. 
	
	 Now we argue that for every $B_j \in \cB$, $|S \cap B_j| \le f(B_j)$. First, note that at most one formula from the disjunction $\bigvee_{Q \in \binom{[k']}{\le t}} \textup{\texttt{exact-nbr}}_{t, Q}(B_j, U) $ can be satisfied, since the disjunction is over mutually exclusive choices for the intersection of a $B_j$ with $U$. Now Consider a $B_j \in \cB_t$, i.e.,$f(B_j) = t$. Then, $\cB_t$ is non-empty. Consider the formula $\psi_t(U)$, and note that $B_j$ satisfies $B_j \in \cB_t$. This implies that the consequent formula of the implication is also satisfied, which, in turn implies that there exists a unique $Q \in \binom{[k']}{\le t}$ such that the formula $\textup{\texttt{exact-nbr}}_{t, Q}(B_j, U)$ is satisfied. Now, let $S_Q = \{ v_i : i \in Q \}$, and note that $S_Q = S \cap B_j$. Therefore, $|B_j \cap S| = |S_Q| = |Q| \le t = f(B_j)$. This shows that $S$ is a fair hitting set for $\cF$.
\end{proof}
\noindent\textbf{Completing the proof of \Cref{thm:nowhere-dense}.} The algorithm iterates over all $2^{\Oh(km)}$ valid pairs $(k', \bbF)$, and uses the algorithm from \Cref{thm:fo-nowhere} to check whether the colored graph $G'$ models the formula $\varphi_{k', \bbF}$. Since the underlying graph $G$ belongs to a nowhere dense class $\cG$, it follows that the model checking algorithm runs in time $f(|\varphi_{k', \bbF}|) n^{\Oh(1)}$ (by setting $\epsilon$ to be some constant, say). If in any iteration the \FO model checking algorithm answers yes, the algorithm concludes that $\cI$ is a yes-instance of \chs. Otherwise, the algorithm returns that $\cI$ is a no-instance of \chs. This concludes the proof of \Cref{thm:nowhere-dense}.

\section{{Parameterization by $k + |\cF|$} when $G_{\cU,\cB}$ is Apex-Minor Free} \label{sec:apex-minor}

An apex graph is a graph $G$ such that for some vertex $v$ (the apex), $G - v$ is planar. 
Let $\cH$ be a minor-closed family of graphs. Then $\cH$ is an apex-minor free if and only if $\cH$ does not contain all apex graphs. In this section, we will design an FPT algorithm for \chs when $G_{\cU, \cB}$ belongs to an apex-minor closed family of graphs, parameterized by $k + |\cF|$. In \Cref{subsec:tw}, we first design an FPT algorithm for \chs parameterized by the treewidth of $G_{\cU, \cF} + |\cF|$. Note that if, in the given instance the $(\cU, \cF)$ graph is $K_{d, d}$ free, then we can use \Cref{thm:kdd-free} to reduce to the case when the size $\cF$ is bounded by $k^{\Oh(d)}$. Then, in \Cref{subsec:baker}, we use the shifting technique to reduce the treewidth of $G_{\cU, \cB}$ to a function of $k$, which enables us to use the result from \Cref{subsec:tw} to obtain the desired result.

\subsection{Parameterization by $\tw(G_{\cU, \cB}) + |\cF|$} \label{subsec:tw}

In this section, we design an FPT algorithm for an \chs instance $(\cU, \cF, \cB, f, k)$, parameterized by the treewidth of the incidence graph $G_{\cU, \cB}$, plus the number of sets in $\cF$. 

For a node $t \in V(T)$, we define the following notation. Let $\cB(t) \coloneqq \cB \cap \beta(t)$, and $\cU(t) \coloneqq \cU \cap \beta(t)$ denote the subsets of $\cB$ and $\cU$ from the bag of $t$, respectively. Similarly, let $\cB_d(t) \coloneqq \cB \cap V(t)$, and $\cU_d(t) \coloneqq \cU \cap V(t)$ denote the subsets of $\cB$ and $\cU$ from the subtree rooted at $t$, respectively. 

For a node $t \in V(T)$, we define a family of functions $\mathcal{G'}(t)$, such that each $g$ is a function with domain $\cB(t)$, and co-domain $\{0, 1, \ldots, k\}$. Then, we define a sub-family of \emph{valid functions} as follows $\mathcal{G}(t) \coloneqq \{g \in \mathcal{G'}(t) : g(B_j) \le f(B_j) \text{ for all } B_j \in \cB(t) \}$.

For every node $t \in V(T)$, each function $g \in \mathcal{G}(t)$, each subset $\cF' \subseteq \cF$, and $\cU' \subseteq \cU(t)$ with $|\cU'| \le k$, we define a table entry $T[t, g, \cF', \cU']$, defined as the size of the smallest $S \subseteq \cU_d(t)$ with the following properties.
\begin{enumerate}
	\item $S$ hits $\cF'$, i.e., for each $F_i \in \cF'$, $S \cap F_i \neq \emptyset$,
	\item For every $B_j \in \cB_d(t)$, $|B_j \cap S| \le f(B_j)$, 
	\item For every $B_{j'} \in \cB(t)$, $|B_{j'} \cap S| = g(B_{j'})$, and 
	\item $S \cap \cU(t) = \cU'$.
\end{enumerate}
\gray{If a set $S \subseteq \cU_d(t)$ (not necessarily of a minimum cardinality) satisfies conditions 1-4 corresponding to an entry $T[t, g, \cF', \cU']$ is said to be a set \emph{matching the description of} the said entry.} If no such $S \subseteq \cU_d(t)$ exists, then the entry is defined to be $\infty$. We observe that for each node $t \in V(T)$, there are at most $(2(k+1))^{\tw+1} \cdot 2^{|\cF|}$ table entries.

Thus, in the following we show how to use dynamic programming to compute table entries $T[t, g, \cF', \cU']$ such that (i) $g(B_j) \le f(B_j)$ for all $B_j \in \cB(t)$, and (ii) $|\cU'| \le k$. Note that since we are interested in solutions of size at most $k$, we may discard entries corresponding to partial solutions of size larger than $k$. However, since this does not affect the computation or the correctness, we do not perform this clean-up during the course of algorithm for the sake of convenience. Finally, in the following we adopt the convention that $\infty + p = \infty$ for any $p \in \mathbb{Z} \cup \{\infty\}$. 

\medskip\noindent\textbf{Leaf node.} Consider a node $t \in V(T)$ without any child in $T$. Then, let $T[t, g, \cF', \cU'] = |\cU'|$ if $|\cU' \cap B_j| \le g(B_j)$ for each $B_j \in \cB(t)$, and $\cU'$ hits $\cF'$. Otherwise, we define $T[t, g, \cF', \cU'] = \infty$.
\\\gray{Since $\cU_d(t) = \cU(t)$, and $\cB_d(t) = \cB(t)$ for a leaf node $t \in V(T)$, it is easy to see that the correctness of the computation.}

\medskip\noindent\textbf{Introduce node.} 
\begin{itemize}
	\item Consider a node $t \in V(T)$ with a child $t' \in V(T)$ such that $\beta(t) = \beta(t') \cup \{u\}$, where $u \in \cU$. In this case,
	\begin{itemize}
		\item If $u \not\in \cU'$, then $T[t, g, \cF', \cU'] = T[t', g, \cF', \cU']$.
		\\\gray{For any $S \subseteq \cU_d(t)$ with $S \cap \cU(t) = \cU'$ with $u \not\in \cU'$, it follows that $S \subseteq \cU_d(t')$. Then, the correctness follows via inductive hypothesis.}
		\item If $u \in \cU'$, then $T[t, g, \cF', \cU'] = 1 + T[t, g', \cF'', \cU'']$, where
		\\$\cF'' = \cF' \setminus \cF(u)$, $\cU'' = \cU' \setminus \{u\}$, and $g'(B_j) = \begin{cases}
			g(B_j) & \text{ if $u \not\in B_j$}
			\\g(B_j) - 1 & \text{ if $u \in B_j$}
		\end{cases}$
		\\\gray{Consider an $S \subseteq \cU_d(t)$ matching the description of $T[t, g, \cF', \cU']$, where $u \in \cU'$. Then, define $S' = S \setminus \{u\} \subseteq \cU_d(t')$. It follows that $|S' \cap B_j| = g'(B_j)$ for any $B_j \in \cB(t')$, which implies that $S'$ matches the description of $T[t', g', \cF'', \cU'']$. Similarly, for any $S'$ matching the description of $T[t', g', \cF'', \cU'']$, $S'\cup \{u\}$ matches the description of $T[t, g, \cF', \cU']$. Thus, the correctness follows from induction.} 
	\end{itemize}
	
	\item Consider a node $t \in V(T)$ with a child $t' \in V(T)$ such that $\beta(t) = \beta(t') \cup \{B_q\}$, where $B_q \in \cB$. In this case,
	\begin{itemize}
		\item If $|\cU' \cap B_q| > g(B_q)$, then $T[t, g, \cF', \cU'] = \infty$.
		\item If $|\cU' \cap B_q| \le g(B_q)$, then $T[t, g, \cF', \cU'] = T[t', g', \cF', \cU']$, where $g'(B_j) = g(B_j)$ for all $B_j \in \cB(t) \setminus \{B_q\}$.
	\end{itemize}
	\gray{In the forward direction, consider a smallest-size $S \subseteq \cU_d(t)$ satisfying the conditions of $T[t, g, \cF', \cU'] \neq \infty$. We claim that $S$ also satisfies the conditions of $T[t', g', \cF', \cU']$. To this end, let $Q = S \cap B_q$. We claim that $Q \subseteq \cU'$. For contradiction, assume that there exists a $u \in Q \setminus \cU'$. First, $u \not\in \cU(t)$, since $S \cap \cU(t) = \cU'$. Therefore, such a $u$ belongs to $\cU_d(t) \setminus \cU(t)$. However, since $u \in B_q$, $\{u, B_q\}$ is an edge in the incidence graph, and $u$ is forgotten in a strict descendant of $t$, this contradicts the properties of a tree decomposition. Combining this with the definition of $g'$, the forward direction follows.
	\\In the backward direction, consider a smallest-size $S' \subseteq \cU_d(t')$ that satisfies (i) the conditions of $T[t', g', \cF', \cU']$, and (ii) $|S' \cap \cU'| \le g(B_q)$. It is easy to see that $S'$ also satisfies the conditions for $T[t, g, \cF', \cU']$. This completes the proof by induction.}
\end{itemize}

\medskip\noindent\textbf{Forget node.} 
\begin{itemize}
	\item Consider a node $t \in V(T)$ with a child $t' \in V(T)$ such that $\beta(t) = \beta(t') \setminus \{u\}$, where $u \in \cU$. In this case, 
	$$T[t, g, \cF', \cU'] = \min \Big\{ T[t', g, \cF', \cU'],\ T[t', g', \cF', \cU''] \Big\},$$
	where $\cU'' = \cU' \cup \{u\}$, and for all $B_j \in \cB(t') = \cB(t)$, $g'(B_j) = \begin{cases}
			g(B_j) & \text{ if $u \not\in B_j$}
			\\g(B_j)+1 & \text{ if $u \in B_j$}
		\end{cases}$
	\gray{Consider an $S \subseteq \cU_d(t)$ that matches the description of $T[t, g, \cF', \cU']$. Let us consider two cases. 
	\begin{itemize}
		\item $u \not\in S$. In this case, we observe that $S$ matches the description of $T[t, g, \cF', \cU']$ if and only if it matches the description of $T[t', g, \cF', \cU']$.
		\item $u \in S$. In this case, note that $\cU'' = S \cap \cU(t') = \cU' \setminus \{u\}$. In this case, we observe that $S$ matches the description of $T[t, g, \cF', \cU']$ if and only if it matches the description of $T[t', g', \cF', \cU'']$. 
	\end{itemize}
	Then, by combining the two cases and using the inductive hypothesis, the correctness follows.
	}
	\item Consider a node $t \in V(T)$ with a child $t' \in V(T)$ such that $\beta(t) = \beta(t') \setminus \{B_q\}$, where $B_q \in \cB$. In this case, let
	$$T[t, g, \cF', \cU'] = \min_{g' \in \mathcal{G}} \Big\{T[t', g', \cF', \cU']\Big\}$$ where $\mathcal{G}$ is a set of functions $g': \cB(t) \to \{0, \ldots, k\}$ satisfying the following two properties: (1) $g'(B_j) = g(B_j) \le f(B_j)$ for all $B_j \in \cB(t') \setminus \{B_q\}$, and $g'(B_q) \ge |B_q \cap \cU'|$.
	\\\gray{Note that in this case, it holds that $\cU(t) = \cU(t')$, and $\cU_d(t) = \cU_d(t')$. Consider an $S \subseteq \cU_d(t)$ matching the description of $T[t, g, \cF', \cU']$. We first claim that $S$ matches the description of $T[t', g', \cF', \cU']$ for some $g \in \mathcal{G}$. To this end, let $Q_1 = \cU' \cap B_q$, and $Q_2 = (S \cap B_q) \setminus Q_1$. Note that $S \cap B_q = S_1 \uplus S_2$. Since $B_q \in \cB_d(t)$, and via the assumption on $S$, it holds that $|S \cap B_q| = |Q_1| + |Q_2| \le f(B_q)$. Therefore, $|Q_1| \le f(B_q) - |S_1|$. We define $g''(B_q) \coloneqq f(B_q) - |Q_1|$ , and for other $B_j \in \cB(t')$, $g''(B_j) = g(B_j)$. It follows that $g'' \in \mathcal{G}$, and $S$ matches the description of $T[t', g'', \cF', \cU']$. 
	\\In the reverse direction, we claim that $|S| \le \min_{g' \in \mathcal{G}} T[t', g', \cF', \cU']$. Note that the proof of forward direction shows that the minimum is finite. Suppose it is realized for some function $g^* \in \mathcal{G}$, and let $S^*$ be a set matching the description of $T[t', g^*, \cF', \cU']$. It is easy to see that $S^*$ also matches the description of $T[t, g, \cF', \cU']$. The correctness follows via inductive hypothesis.
	}
\end{itemize}

\medskip\noindent\textbf{Join node.} 
Consider a node $t \in V(T)$ with two children $t_1, t_2 \in V(T)$, such that $\beta(t) = \beta(t_1) = \beta(t_2)$. Fix an entry $T[t, g, \cF', \cU']$. Before discussing how to compute this entry, we define some notation. Let $\cF_{\text{good}}$ be a set of pairs $(\cF_1, \cF_2)$, where $\cF_1, \cF_2 \subseteq \cF'$ are sub-familites satisfying the following two conditions: (i) $\cF_1 \cup \cF_2 = \cF'$, and (ii) $\bigcup_{u \in \cU'} \cF(u) \subseteq \cF_1 \cap \cF_2$. Furthermore, let $\mathcal{G}_{\text{good}}$ be a set of pairs $(g_1, g_2)$ of functions, where $g_1 \in \mathcal{G}(t_1)$, and $g_2 \in \mathcal{G}(t_2)$ that satisfy the following property: for every $B_j \in \cB(t)$: $g_1(B_j) + g_2(B_j) = g(B_j) + |\cU' \cap B_j|$.
\begin{align}
	T[t, g, \cF', \cU'] &= \min_{\substack{(\cF_1, \cF_2) \in \cF_{\text{good}},\\ (g_1, g_2) \in \mathcal{G}_{\text{good}}}} \Big\{ T[t_1, g_1, \cF_1, \cU'] + T[t_2, g_2, \cF_2, \cU'] - |\cU'| \Big\} \label{eqn:tw-join}
\end{align}
\gray{From the properties of tree decompositions, it holds that $\cU_d(t) = \cU_d(t_1) \cup \cU_d(t_2)$, and $\cU_d(t_1) \cap \cU_d(t_2) = \cU(t)$. Consider an $S \subseteq \cU_d(t)$ matching the description of $T[t, g, \cF', \cU']$, and let $S_1 = S \cap \cU_d(t_1)$, and $S_2 = S \cap \cU_d(t_2)$. Note that $S_1 \cap S_2 =  S \cap \cU(t) = S_1 \cap \cU(t_1) = S_2 \cap \cU(t_2) = \cU'$. This implies that $|S| = |S_1| + |S_2| - |\cU'|$.

Let $\cF'_1$ and $\cF'_2$ be the subsets of $\cF$ that are hit by $S_1$ and $S_2$ respectively. It is easy to see that $(\cF'_1, \cF'_2)$ is a \emph{good pair}. Similarly, let $g'_1 \in \mathcal{G}(t_1)$, $g'_2 \in \mathcal{G}(t_2)$ be two functions, such that for every $B_j \in \cB(t_1) = \cB(t_2)$, $g'_1(B_j) = |S_1 \cap B_j|$, and $g'_2(B_j) = |S_2 \cap B_j|$. Note that $|S \cap B_j| = |S_1 \cap B_j| + |S_2 \cap B_j| - |\cU' \cap B_j|$, i.e., $g(B_j) = g'_1(B_j) + g'_2(B_j) - |\cU' \cap B_j|$. Therefore, $(g'_1, g'_2) \in \mathcal{G}_{\text{good}}$. Thus, $S_1$ and $S_2$ match the descriptions of $T[t_1, g_1, \cF_1, \cU']$, and $T[t_2, g_2, \cF_2, \cU']$, respectively. This implies that $|S| \ge T[t_1, g'_1, \cF'_1, \cU'] + T[t_2, g'_2, \cF'_2, \cU'] - |\cU'|$, which shows the forward direction. 

In the reverse direction, we first observe that the proof of the forward direction implies that the minimum of Equation \ref{eqn:tw-join} is finite. Let $(\cF^*_1, \cF^*_2) \in \cF_{\text{good}}$, and $(g^*_1, g^*_2) \in \mathcal{G}_{\text{good}}$ be the pairs realizing a minimum. Additionally, let $S^*_1 \subseteq \cU_d(t_1)$ and $S^*_2 \subseteq \cU_d(t_2)$ be the sets matching the descriptions of $T[t_1, g^*_1, \cF^*_1, \cU']$ and $T[t_2, g^*_2, \cF^*_2, \cU']$ respectively. Then, it is easy to see that $S^* = S^*_1 \cup S^*_2$ matches the description of $T[t, g, \cF', \cU']$, which implies that $|S| \le |S^*|$. This finishes the proof by induction.
}

It is easy to see that each table entry can be computed in $k^{\Oh(\tw)} \cdot 2^{\Oh(|\cF|)} \cdot |\cI|^{\Oh(1)}$ time. Finally, to compute the final answer, we examine the entries $T[r, \cdot, \cF', \cdot]$ corresponding to the root bag $r \in V(T)$ such that $\cF' = \cF$, i.e., the partial solution hits all sets in $\cF$. If any such entry is a finite integer that is at most $k$, we conclude that the given instance admits a solution of size at most $k$. Otherwise, if each such entry is either infinite, or is an integer larger than $k$, then we conclude that there is no solution of size at most $k$, i.e., that the given instance is a no-instance. Thus, we conclude with the following theorem.

\begin{theorem}\label{thm:tw-fhs}
	Given an instance $\cI = (\cU, \cF, \cB, f, k)$, and a nice tree-decomposition of the incidence graph $G_{\cU, \cB}$ of width $\tw$, there exists an $k^{\Oh(\tw)} \cdot 2^{\Oh(|\cF|)} \cdot |\cI|^{\Oh(1)}$ time algorithm to decide whether $\cI$ is a yes-instance of \chs.
\end{theorem}

\subsection{The case when $G_{\cU, \cF}$ is Apex-Minor Free} \label{subsec:baker}


For a vertex $v$ of a graph $G$ and an integer $r \geq 0$, by $G_{v}^{r}$ we denote the subgraph of $G$ induced by vertices within a distance at most $r$ from $v$ in $G$. We already know the dependence of treewidth on $r$ in a apex-minor free graph from some of the following known results, e.g., \cite{Eppstein2000}.

\begin{theorem}{\label{thm:diametertreewidthforapexminorfree}}
	Let $\cH$ be a family of apex-minor free graphs, and $G \in \cH$. Let $v \in V(G)$ be an arbitrary vertex, and $r$ be a nonnegative integer.
	Then $\tw(G_{v}^{r}) \leq g(r)$, for some function $g$.
	Moreover, decomposition of $G_{v}^{r}$ of width at most $g(r)$ can be constructed in polynomial time.
\end{theorem}

By \cref{thm:diametertreewidthforapexminorfree}, we have the following corollary, the proof of which is analogous to the proof of Corollary 7.34 in \cite{CyganFKLMPPS15}.

\begin{corollary} \label{cor:treewidthofthesubgraph}
	Let $v$ be a vertex of an apex-minor free graph $G$, and $i \geq 0$.
	Let $L_i$ be the set of vertices of $G$ that are at distance exactly $i$ from $v$.
	Then for any $i,j \geq 1$, the treewidth of the subgraph $G_{i,i+j-1}=G[L_i\cup L_{i+1}\cup \ldots \cup L_{i+j-1}]$ does not exceed $g(r)$.
	Moreover, a tree decomposition of  $G_{i,i+j-1}$ of width at most $g(r)$ can be computed in polynomial time.
\end{corollary}
We now consider an instance $\cI=(\cU,\cF, \cB,f,k)$ of a \chs such that $G=G_{\cU,\cB} \in \cH$.  
We will design an FPT algorithm for this case. 
\begin{definition}
A partition $\cP=\{\cU_1,\ldots,\cU_{\ell}\}$ of $\cU$ in $G=G_{\cU,\cB}$ is called a $g$-\emph{treewidth bounded partition} of $\cU$ , if for given an integer $k \geq 0$ and  some function $g$ it satisfies the following property: for any $i \in [k+1]$, the graph  $G[(\cU\setminus \cU_i) \cup \cB]$ has treewidth at most $g(k)$.
\end{definition}

We now prove the following result.
\begin{lemma}{\label{lem:apexminorhasboundedtreewidth}}
	Let $\cH$ be an apex-minor closed family, and let $G \in \cH$, then there exists a function $g$, such that  $G$ admits a $g$-treewidth bounded partition.
	Moreover, such a partition, together with the corresponding tree decompositions of width at most $g(k)$ can be found in polynomial time.
\end{lemma}

\begin{proof} Without loss of generality, we assume that $G$ is connected -- otherwise we can apply the following procedure for each connected component, and take the union of the corresponding $g$-treewidth bounded partitions for each connected component.

We select an arbitrary vertex $v \in \cU$ and run a breadth first search (BFS) from $v$.
Let $L_1,L_2, \ldots$ be the layers of the BFS tree.
Observe that all the odd layers contain exactly the elements of $\cU$ and the even layers comprise of elements of $\cB$.
We partition $\cU=\{\cU_1,\ldots,\cU_{k+1}\}$, where 
$\cU_i=\bigcup_{ a\geq 0}^{}L_{a(2k+2)+2i-1}$, for $i \geq 1$.
Note that, using the notation from the statement of \Cref{cor:treewidthofthesubgraph}, the graph $G[\cU \setminus \cU_i,\cB ]$ is a disjoint union of graphs $G_{a(2k+2)+2i,(a+1)(2k+2)+2i-2}$ for $a \ge 0$ and $G_{v}^{i-1}$, for $i>0$.
Then by \cref{thm:diametertreewidthforapexminorfree} and \Cref{cor:treewidthofthesubgraph}, for each of these graphs we can construct a tree decomposition of width at most $g(k)$; connecting these tree decompositions arbitarily yields a tree decomposition of $G[\cU \setminus \cU_i,\cB ]$ of width at most $g(k)$.
Hence, partition $\LR{\cU_i}_{1 \leq i\leq k+1}$  is $g$-treewidth bounded partition of $\cU$ in $G$.
\end{proof}

In the following theorem, we use the $g$-treewidth bounded partition of $G$ to reduce the problem to an instance where $G_{\cU, \cB}$ incidence graph has bounded treewidth. 

\begin{theorem}{\label{thm:deletionofpartition}}
	Let $\cH$ be a apex-minor free graph class.
	Let $\cI=(\cU,\cF,\cB,f,k)$ be an instance of \chs such that $G:=G_{\cU,\cB} \in \cH$. 
	Let $\cP=\LR{\cU_1,\ldots,\cU_\ell}$ be a $g$-treewidth bounded partition of $\cU$.
	Then $\cI$ is a yes-instance of \chs if and only if there exists some $j \in [k+1]$ such that instance $\cI_j=(\cU\setminus \cU_j,\cF-\cU_j ,\cB- \cU_j,f,k)$ is a yes-instance of \chs.
\end{theorem}
\begin{proof}
	In the forward direction. 
	Assume $X \subseteq \cU$ is a fair hitting set for $\cF$. 
	Observe  $|X|\leq k$,
	thus for some $j \in [k+1]$, there exists a partition set $\cU_j$ such that $X \cap \cU_j=\emptyset$.
	Clearly, $X \subseteq \cU\setminus \cU_j$, and note that it is a hitting set for $\cF'=\cF - \cU_j$ of size at most $k$.
	Also, for every $B'_i \in \cB'=\cB - \cU_j$, we have  $|X \cap B'_i | \le f(B_i)$.
	Hence, $X$ is a fair hitting set for $\cF'$.
	
	In the reverse direction, suppose $X \subseteq \cU \setminus \cU_j$ is a fair hitting set for $\cF'$ of size at most $k$.
	We claim that $X$ is also a fair hitting set for $\cF$ of size at most $k$.
Since set $X$ hits $\cF'$ and there is a bijection between the elements of set $\cF'$ and $\cF$, therefore, $X$ also hits $\cF$. Again, since there is a bijection between the elements of $\cB'$ and $\cB$ and $f's$ value remain unchanged, it follows that for every $B_i\in \cB$, we have $|X \cap B_i | \leq f(B_i)$. Hence, $X$ is also a fair hitting set for $\cF$.
	\end{proof}

Thus, our algorithm uses \Cref{lem:apexminorhasboundedtreewidth} to compute a $g$-treewidth bounded partition of $\cU$. Note that for any $j \in [k+1]$, the treewidth of graph $G[(\cU \setminus \cU_j) \cup \cB]$ is upper bounded by $g(k)$. Furthermore, using \Cref{thm:deletionofpartition}, it follows that $\cI$ is a yes-instance iff there exists a $j \in [k+1]$ such that $\cI_j$ is a yes-instance of \chs. Thus, the algorithm uses \Cref{thm:tw-fhs} to solve each instance $\cI_j$, and if any such instance is a yes-instance, it concludes that $\cI$ is a yes-instance. Otherwise, the algorithm concludes that $\cI$ is a no-instance. Thus, we obtain the following theorem.

\begin{restatable}{theorem}{apexm}\label{thm:tw-apex-minorfree}
	Given an instance $\cI = (\cU, \cF, \cB, f, k)$ such that $G_{\cU, \cB}$ belongs to a family of apex-minor free graphs, and $G_{\cU, \cF}$ is $K_{d, d}$-free, there exists an $g(k, d) \cdot |\cI|^{\Oh(1)}$ time algorithm to decide whether $\cI$ is a yes-instance of \chs.
\end{restatable}
\begin{remark}
	Note that if $\cH$ is the family of planar graphs, then the function $g(k)$ from \Cref{lem:apexminorhasboundedtreewidth} is $\Oh(k)$. Thus, in this case, the running time of the algorithm is $k^{\Oh(k)} \cdot 2^{|\cF|} \cdot |\cI|^{\Oh(1)}$.
\end{remark}
\section{Lower Bounds} \label{sec:lb}

First, we observe that \textsc{Hitting Set} is a special case of \chs, when $\cB = \emptyset$. Thus, the problem is \np-hard, and W$[1]$-hard parameterized by the solution size $k$. 

In the following two subsections, we show two lower bound results for \chs that hold in special cases of interest. First, in \Cref{subsec:nphard}, we show that \chs remains \np-hard even in a very special case. Next, in \Cref{subsec:w1hard}, we show that \chs remains W$[1]$-hard, again, in a very special case.

\subsection{$\np$-hardness for a Special Case of \chs} \label{subsec:nphard}

In this section, we show that \chs remains \np-hard even in the special case when the sets in $\cF$ are disjoint, and each element of the universe is contained in at most $2$ sets of $\cB$. To this end, we reduce from the following problem, shown to be \np-hard in \cite{BandyapadhyayFIS23}.

\begin{tcolorbox}[colback=white!5!white,colframe=gray!75!black]
	\textsc{Exact Rainbow Matching}
	\\\textbf{Input.} A proper edge-colored graph $G = (V, E)$ on $n$ vertices, where each $uv \in E(G)$ is colored with a color from $\{1, \ldots, q\}$. 
	\\\textbf{Question.} Does there exist a matching $M \subseteq E(P)$ such that $M$ contains \emph{exactly} one edge of each color $c$, $1 \le c \le q$?
\end{tcolorbox}

\begin{proposition}{(Theorem 5 from \cite{BandyapadhyayFIS23})}
	\textsc{Exact Rainbow Matching} is \np-complete even on $n$-vertex properly edge-colored paths, such that the number of colors is $\Theta(n)$. Furthermore, assuming ETH, the problem does not admit a $2^{o(n)}$ algorithm in the same setting.
\end{proposition}

\medskip\noindent\textbf{Reduction.} Let $P$ be a path that is input to \textsc{Fair Matching on a Path}. For $1 \le c \le q$, let $E_c$ denote the set of edges colored with $c$. Note that $\{E_1, E_2, \ldots, E_q\}$ is a partition of $E(P)$, and we can assume that each $E_c$ is non-empty (otherwise we can immediately conclude that it is a no-instance). Let us also assume that the vertices of the path are numbered from left to right as $v_1, v_2, \ldots, v_n$, and $e_i = v_{i}v_{i+1} \in E(P)$ for $1 \le i \le n-1$. 

Now we construct an equivalent instance of \chs as follows. For each edge $e_i$, we add an element $u_i$ to the universe. Thus, $\cU= \{u_1, \ldots, u_{n-1}\}$. For each color $1 \le c \le q$, we add a set $F_c$ that contains the elements corresponding to the edges in $E_c$. For each $1 \le i \le n-1$, we add a set $B_i = \{u_i, u_{i+1}\}$ to $\cB$, and set $f(B_i) = 1$ for all $B_i \in \cB$. Finally, we set $k = q$, and let $\cI = (\cU, \cF, \cB, f, k)$ be the resulting instance of \chs. Note that the sets in $\cF$ are disjoint, and the frequency of any element of $\cU$ in the sets of $\cB$ is at most two. 

\medskip\noindent\textbf{Equivalence.} In the forward direction, suppose $P$ admits a matching $M \subseteq E(P)$ that contains exactly one edge of each color. Then, let $S \subseteq \cU$ be the set of elements corresponding to the edges of $M$. We claim that $S$ is a \chs for $\cF$ of size $k = q$. First, since $M$ contains exactly one edge of each color, $|S| = |M| = q = k$, and it also follows that $S$ hits every set in $\cF$. Finally, since $M$ is a matching, it does not contain two edges incident to the same vertex. It follows that $|S \cap B_i| \le 1$ for every $B_i \in \cB$. This completes the forward direction.

In the reverse direction, let $S$ be a fair hitting set for $\cF$ of size at most $k$, and let $M$ be the set of edges corresponding to the elements of $S$. Since $S$ hits every set in $\cF$, $M$ contains exactly one edge of each color. Furthermore, since $S$ hits every set in $\cF$, $|S| = |\cF| = k$, which means that $M$ contains \emph{exactly} one edge of each color. Finally, since $|S \cap B_i| \le 1 = f(B_i)$ for every $B_i \in \cB$, it follows that $M$ does not contain two edges incident to a particular vertex, which means that $M$ is a matching. Thus, we conclude with the following theorem.

\begin{theorem} \label{thm:nph}
	\chs is \np-hard, and does not admit a $2^{o(n)}$ time algorithm, even in the following special case, when the instance $\cI = (\cU, \cF, \cB, f, k)$ satisfies following properties.
	\begin{itemize}
		\item The sets in $\cF$ are pairwise disjoint,
		\item The sets in $\cB$ have size exactly $2$,
		\item Each $u \in \cU$ belongs to at most $2$ different sets of $\cB$,
		\item The incidence graph $G_{\cU, \cB}$ is a path, and
		\item $|\cB| = n-2, |\cU| = n-1,$ and  $|\cF| = k = \Theta(n)$.
	\end{itemize}
\end{theorem}

\subsection{$\textup{W}[1]$-hardness of \chs} \label{subsec:w1hard}
In this section, we show that \chs is W[1]-hard parameterized by $k$. This result holds when the the family $\cF$ is disjoint, but the sets in $\cB$ can intersect arbitrarily. 

We establish the result for the fixed-parameter intractability of $\chs$ assuming sets in $\cB$ are arbitrary.
For this, we give a parameter-preserving reduction from the $W[1]$-complete $\kmci$ to $\chs$ parameterized by $k$.
In \kmci, given an undirected graph $G$, an integer $k$, and $(V_1 \ldots,V_k)$ a partition of $V(G)$.
Decide whether $G$ has a set $I \subseteq V (G)$ such that $|I \cap V_i|=1$ for every $i \in  [k]$, and $G[I]$ is edgeless.

Formally, we prove the following result.

\begin{theorem} \label{thm:w1hard}
	$\chs$ is \textup{W}$[1]$-hard when parameterized by $k$, assuming the sets in $\cB$ are arbitrary.
\end{theorem}
\begin{proof}
	Let $(G, k,V_1 \ldots V_k)$ be an instance of $\kmci$.
	We construct an instance $(\cU, \cF, \cB, f, k)$ as follows.
	\begin{itemize}
		\item For every vertex $v \in V(G)$, we add an element $e(u)$ to  $\cU$.
		\item For every $1\leq i \leq k$, we make the set $F_i=\{e(u) : u \in V_i\}$. Also set $\cF:=\{F_1,\ldots,F_k\}$.
		\item For every edge $\{u,v\} \in E(G)$, we introduce a set $V_{\{u,v\}}:=\{e(u),e(v)\}$.
		Set $\cB :=\cup_{ \{u,v\} \in E(G)}^{}V_{\{u,v\}}$ and $f(V_{\{u,v\}})=1$.
		
	\end{itemize}
	We claim that $G$ has an independent set with exactly one vertex from each $V_i$ if and only if $\cU$ has a fair hitting set for $\cF$ of size $k$.
	
	In the forward direction, suppose there is an independent set $I$ in $G$ with $|I \cap V_1|=|I \cap V_2|= \ldots =|I \cap V_k|=1$
that is, one vertex from each $V_i$. 
We show that $I$ is a fair hitting set for  $\cF$.
	Clearly, $I \cap F_i \neq \emptyset$ for each $i \in \{1,\ldots,k\}$, as $I \cap V_i \neq \emptyset$.
	Moreover, $|I \cap F_i|=|I \cap V_i|=1$ and $|\cF|=k$,
	thus, $I$ is a hitting set for $\cF$ of size $k$.
	We next claim that for every set $V_{\{u,v\}} \in \cB$,  $|I \cap V_{\{u,v\}}| \leq f(V_{\{u,v\}})=1$.
	Assume to the contrary, that there exists a set $V_{\{u',v'\}} \in \cB$ such that  $|I \cap V_{\{u',v'\}}| > 1$. 
	This implies that $I$ contains both the endpoints $u'$ and $v'$ of an edge, contradicts the assumption that $I$ is an independent set in $G$.
	
	In the reverse direction, suppose $S \subseteq \cU$ is a fair hitting set for $\cF$ of size $k$.
	Denote by $S'$ the set of vertices in $G$ corresponding to the elements of $S$ in $\cU$.	
	First, as $|\cF|=k$ and $S \cap F_i \neq \emptyset$, for each $i \in \{1,\ldots, k\}$ implies that $|S' \cap V_i|=1$.
	We now claim that $S'$  forms an independent set in $G$.
	Suppose not, this means there exists an edge $\{u,v\} \in E(G)$ such that $|S' \cap \{u,v\}|=2$.
	This implies that $|S \cap V_{\{u,v\}}|=|S \cap \{e(u),e(v)\}|=2 >1$, contradicting the assumption that $|S \cap V_{\{u,v\}}| \leq f(V_{\{u,v\}})=1$, for all $\{u,v\} \in E(G)$.
\end{proof}

Note that it is necessary in the construction that the $B_i \in \cB$ can intersect arbitrarily. Nevertheless, we observe that in the constructed instance of \chs, it holds that the sets in $\cF$ are disjoint, and $|B_i| = 2$ for all $B_i \in \cB$, which also implies that the graph $G_{\cU, \cB}$ is $K_{2, 2}$-free, and even $2$-degenerate. This shows the limit of fixed-parameter tractability of \chs.

\bibliography{ref}

\begin{thebibliography}{10}

\bibitem{BandyapadhyayFIS23}
{\sc S.~Bandyapadhyay, F.~V. Fomin, T.~Inamdar, and K.~Simonov}, {\em
  Proportionally fair matching with multiple groups}, 2023.

\bibitem{BlumDFGP22}
{\sc J.~Blum, Y.~Disser, A.~E. Feldmann, S.~Gupta, and A.~Zych{-}Pawlewicz},
  {\em On sparse hitting sets: from fair vertex cover to highway dimension},
  CoRR, abs/2208.14132 (2022).

\bibitem{CyganFKLMPPS15}
{\sc M.~Cygan, F.~V. Fomin, L.~Kowalik, D.~Lokshtanov, D.~Marx, M.~Pilipczuk,
  M.~Pilipczuk, and S.~Saurabh}, {\em Parameterized Algorithms}, Springer,
  2015.

\bibitem{DiestelGT}
{\sc R.~Diestel}, {\em Graph Theory, 4th Edition}, vol.~173 of Graduate texts
  in mathematics, Springer, 2012.

\bibitem{Eppstein2000}
{\sc D.~Eppstein}, {\em Diameter and treewidth in minor-closed graph families},
  Algorithmica, 27 (2000), pp.~275--291.

\bibitem{FominLPS16}
{\sc F.~V. Fomin, D.~Lokshtanov, F.~Panolan, and S.~Saurabh}, {\em Efficient
  computation of representative families with applications in parameterized and
  exact algorithms}, J. {ACM}, 63 (2016), pp.~29:1--29:60.

\bibitem{GroheKS17}
{\sc M.~Grohe, S.~Kreutzer, and S.~Siebertz}, {\em Deciding first-order
  properties of nowhere dense graphs}, J. {ACM}, 64 (2017), pp.~17:1--17:32.

\bibitem{JainKM20}
{\sc P.~Jain, L.~Kanesh, and P.~Misra}, {\em Conflict free version of covering
  problems on graphs: Classical and parameterized}, Theory Comput. Syst., 64
  (2020), pp.~1067--1093.

\bibitem{karp1972}
{\sc R.~M. Karp}, {\em Reducibility among combinatorial problems}, in
  Complexity of computer computations, Springer, 1972, pp.~85--103.

\bibitem{KnopMT19}
{\sc D.~Knop, T.~Masar{\'{\i}}k, and T.~Toufar}, {\em Parameterized complexity
  of fair vertex evaluation problems}, in 44th International Symposium on
  Mathematical Foundations of Computer Science, {MFCS} 2019, August 26-30,
  2019, Aachen, Germany, P.~Rossmanith, P.~Heggernes, and J.~Katoen, eds.,
  vol.~138 of LIPIcs, Schloss Dagstuhl - Leibniz-Zentrum f{\"{u}}r Informatik,
  2019, pp.~33:1--33:16.

\bibitem{LinS89a}
{\sc L.~Lin and S.~Sahni}, {\em Fair edge deletion problems}, {IEEE} Trans.
  Computers, 38 (1989), pp.~756--761.

\bibitem{LokshtanovMPS18}
{\sc D.~Lokshtanov, P.~Misra, F.~Panolan, and S.~Saurabh}, {\em Deterministic
  truncation of linear matroids}, {ACM} Trans. Algorithms, 14 (2018),
  pp.~14:1--14:20.

\bibitem{LokshtanovPSSZ20}
{\sc D.~Lokshtanov, F.~Panolan, S.~Saurabh, R.~Sharma, and M.~Zehavi}, {\em
  Covering small independent sets and separators with applications to
  parameterized algorithms}, {ACM} Trans. Algorithms, 16 (2020),
  pp.~32:1--32:31.

\bibitem{Marx09}
{\sc D.~Marx}, {\em A parameterized view on matroid optimization problems},
  Theor. Comput. Sci., 410 (2009), pp.~4471--4479.

\bibitem{MasarikT20}
{\sc T.~Masar{\'{\i}}k and T.~Toufar}, {\em Parameterized complexity of fair
  deletion problems}, Discret. Appl. Math., 278 (2020), pp.~51--61.

\end{thebibliography}

\end{document}